\RequirePackage{fix-cm}
\documentclass[smallextended]{svjour3}       
\smartqed  
\usepackage{graphicx}
\usepackage{amssymb}
\begin{document}

\title{Efficient constant factor approximation algorithms for stabbing line segments with equal disks
\thanks{This work was supported by Russian Foundation for Basic Research, project 19-07-01243.
The paper represents a major extension of the short paper appeared in the proceedings of the 12 Annual
International Conference on Learning and Intelligent Optimization (LION 2018) \cite{kobylkin_lion}. It also extends work presented in the
International Conference on Mathematical Optimization Theory and Operations Research (MOTOR 2019)
\cite{kobylkin_dryakhlova}.}}

\author{Konstantin Kobylkin}

\titlerunning{Efficient approximation algorithms for stabbing line segments with equal disks}        


\institute{K. Kobylkin \at
              Krasovsky Institute of Mathematics and Mechanics, Russia, Ekaterinburg, Sophya Kovalevskaya str. 16\\
              Ural Federal University, Russia, Ekaterinburg, Mira str. 19\\
              \email{kobylkinks@gmail.com}           
}

\date{Received: date / Accepted: date}

\maketitle
\begin{abstract}
An NP-hard problem is considered of
intersecting a given set of $n$ straight line segments on the plane
with the smallest cardinality set of disks of fixed radii $r>0,$ where the set of
segments forms a straight line drawing $G=(V,E)$ of a planar graph without proper edge
crossings. To the best of our knowledge, related work only tackles a setting where $E$ consists of
(generally, properly overlapping) axis-parallel segments,
resulting in an $O(n\log n)$-time and $O(n\log n)$-space 8-approximation algorithm.
Exploiting tough connection of the problem with the
geometric Hitting Set problem, an $\left(50+52\sqrt{\frac{12}{13}}+\nu\right)$-approximate $O\left(n^4\log n\right)$-time and
$O\left(n^2\log n\right)$-space algorithm is devised based on the modified Agarwal-Pan algorithm, which uses epsilon nets.
More accurate $(34+24\sqrt{2}+\nu)$- and
$\left(\frac{144}{5}+32\sqrt{\frac{3}{5}}+\nu\right)$-approxi\-mate
algorithms are also proposed for cases where $G$ is any subgraph of
either a generalized outerplane graph or a Delaunay triangulation respectively,
which work within the same time and space complexity bounds, where $\nu>0$ is an arbitrarily small constant.


\end{abstract}

\section{Introduction}
\label{intro}

Design of fast and accurate approximation algorithms is a very important topic in the class of
combinatorial optimization problems, being both NP- and W[1]-hard. Roughly speaking, a problem W[1]-hardness means that
every polynomial time approximation scheme (PTAS) to solve the problem must have complexity of the order
$\Omega\left(L^{f(1/\varepsilon)}\right)$ for a monotonic computable function $f$ if some reasonable conjecture holds true, where
$L$ represents length of the problem input. Many problems from computational geometry
are both NP- and W[1]-hard, including a wide class of problems related to
optimal coverage, piercing or intersection of given families of geometric objects on the plane with
simply shaped objects. Some problems from this class
can be stated in the following general form. Suppose a family ${\cal{F}}$ is given of objects from $\mathbb{R}^2$
of constant complexity, which can be e.g. disks, straight line segments, triangles etc.
The problem is to find the smallest cardinality set ${\cal{D}}$ of translates of a given object $D_0\subset\mathbb{R}^2$
such that either for each $F\in {\cal{F}}$ there is some $D=D(F)\in{\cal{D}},$ which intersects $F$ in some prescribed way,
or $F\subset\bigcup\limits_{D\in{\cal{D}}}D$ for each $F\in {\cal{F}}.$
Design and analysis of approximation algorithms for problems, which can be written in this form, is an area of ongoing research
(see e.g. works \cite{basappa}, \cite{biniaz}, \cite{nandy}, \cite{kobylkin_dryakhlova},
\cite{mudgal}).

In the present paper fast constant factor approximation algorithms are
constructed for a special NP-hard (\cite{nandy}, \cite{kobylkin})
geometric intersection problem, which fits in the general form above. Namely, in this
problem, ${\cal{F}}$ coincides with a set $E$ of straight line segments, ${\cal{D}}$ consists of identical
disks and object intersection is understood in its common sense.
The problem has its applications in
facility location and sensor network deployment.

\noindent {\sc Intersecting Plane Graph with Disks (IPGD)}:
given a straight line drawing (or a plane graph) $G=(V,E)$
of an arbitrary simple planar graph without proper edge crossings
and a constant $r>0,$ find the smallest cardinality
set ${\cal{D}}$ of disks of radius $r$
such that $e\cap\bigcup\limits_{D\in{\cal{D}}}D\neq\varnothing$
for each edge $e\in E.$ Here each isolated vertex $v\in V$ is treated as a zero-length segment
$e_v\in E.$

This problem is obviously equivalent to finding the smallest cardinality point set
$C\subset\mathbb{R}^2$ such that each $e\in E$ is within Euclidean distance $r$
from some point $c=c(e)\in C.$ In fact, $C$ represents a set of centers of radius $r$
disks, forming a solution to the {\sc IPGD} problem.

The latter
equivalent formulation of the {\sc IPGD} problem can further be reduced
to a special case of the well known
geometric {\sc Hitting Set} problem on the plane. In its
general form the {\sc Hitting Set} problem is formulated as follows:

\noindent {\sc Hitting Set}: Given a point set $Y\subseteq\mathbb{R}^2$
and a family ${\cal{R}}$ of subsets (also called {\it objects}) of
$\mathbb{R}^2,$ find the smallest cardinality subset $H\subseteq Y$ such
that $H\cap R\neq\varnothing$ for all $R\in{\cal{R}}.$ Here,
a subset $H\subset \mathbb{R}^2$ is named a {\it hitting set} for ${\cal{R}}$
if $H\cap R\neq\varnothing$ for every $R\in{\cal{R}}.$

A pair $(Y,{\cal{R}})$ is associated with each instance of the {\sc Hitting Set} problem,
which is called a {\it range space}.

To describe a {\sc Hitting Set} formulation of the {\sc IPGD} problem, some notation
is given below.
Suppose
$N_r(e)=\{x\in\mathbb{R}^2:d(x,e)\leq r\},$ ${\cal{N}}_r(E)=\{N_r(e):e\in E\}$
and $d(x,e)$ is Euclidean distance between a point $x\in{\mathbb{R}}^2$
and a segment $e\in E;$ for a zero-length segment $x\in\mathbb{R}^2$
$N_r(x)$ denotes a radius $r$ disk centered at $x.$
Each object from ${\cal{N}}_r(E)$ is a Euclidean $r$-neighborhood of some segment
of $E$ also called {\it $r$-hippodrome} or {\it $r$-offset} in the literature \cite{nandy}.

The {\sc IPGD} problem can equivalently be  formulated as follows:

{\sc Piercing Euclidean Hippodromes (PEH).} Given
a set ${\cal{N}}_r(E)$ of $r$-hippodromes on the plane whose underlying straight line segments
form an edge set of some plane graph $G=(V,E),$
find the minimum cardinality hitting set for ${\cal{N}}_r(E).$

Thus, the {\sc IPGD} problem is reduced to the {\sc Hitting Set} problem with $Y=\mathbb{R}^2$
and ${\cal{R}}={\cal{N}}_r(E),$ where ${\cal{N}}_r(E)$ is a set of $r$-hippodromes,
formed by segments from $E.$

\subsection{Applications}

The {\sc IPGD} problem is of interest in network security analysis, sensor network deployment and
facility location. In the work \cite{nandy} its sensor deployment applications are
reported for road networks. Namely, in this work an application is considered to monitor
a road network using identical sensors
with circular sensing areas. Geometrically, network roads are modeled by piecewise linear arcs
on the plane. One can split these arcs into chains of elementary straight line segments
such that any two of the resulting elementary segments intersect at most
at their endpoints. When full road network surveillance is costly,
it might be a good idea to place the minimum number of sensors
such that each piece of every road (represented by an elementary segment)
is partially covered by sensing area
of some of the placed sensors. Using this network of the placed sensors,
geographic locations of all those vehicles can be identified which move on the road network.
Here coordinates of moving vehicles are given by coordinates of the respective elementary segments.
The aforementioned modeling approach leads to a geometric combinatorial
optimization model, which coincides with the {\sc IPGD} problem.

A variety of settings of optimal location problems to place facilities (e.g. petrol stations)
nearby a given network roads can be reduced to the {\sc IPGD} problem.
Another network security analysis application may also be of interest of the {\sc IPGD} problem
for optical fiber networks, which is inspired by work of
\cite{journals/ton/AgarwalEGHSZ13} and described in detail in \cite{kobylkin}.

\subsection{Related work}

To the best of our knowledge,
settings close to the {\sc IPGD} problem are first considered
in \cite{nandy}.
They explore the case in which ${\cal{D}}$ contains identical disks
and $E$ consists of (generally properly overlapping)
axis-parallel segments. Their algorithms
can easily be extended to the case of sets $E$ of straight line segments with bounded number of
distinct orientations. Moreover, in \cite{nandy} two polynomial time approximation schemes are proposed.
The first one is for the case
of axis-parallel segments and another one is for the case of segments, whose Euclidean lengths are
within a constant factor of $r.$

The {\sc IPGD} problem generalizes a classical NP-hard unit disk covering problem.
In the unit disk covering problem one needs to cover a given finite point set $E$
on the plane with the least cardinality set ${\cal{D}}$ of unit disks. In the {\sc IPGD}
problem setting $E$ generally contains non-zero length segments instead of points.

A PTAS exists \cite{mudgal} based on local search for more general version of
the {\sc IPGD} problem in which disks of ${\cal{D}}$ are chosen from some prescribed
finite set ${\cal{H}}$ of generally non-equal disks. Existence of PTAS can also be
established by reduction of the {\sc IPGD} problem to the {\sc PEH}
problem. More precisely, the fact can be proved that
${\cal{N}}_r(E)$ is a family of closed convex pseudo-disks under some mild restrictions on $E,$
and, then, a PTAS is constructed for the {\sc PEH} problem, using the general approach
from \cite{muray}.

To emphasize a close connection of the {\sc IPGD} problem with the {\sc Hitting Set} problem
a proof sketch is given below of the latter fact.
More precisely, without loss of generality
it can be shown, according to the definition
of pseudo-disks, that $|{\rm{bd}}\,N_1\cap {\rm{bd}}\,N_2|\leq 2$
for any distinct $N_1,N_2\in {\cal{N}}_r(E)$ and both $N_1\backslash N_2$
and $N_2\backslash N_1$ are connected, where ${\rm{bd}}\,N$ denotes
boundary of a set $N\subset\mathbb{R}^2$ and $|N|$ denotes its cardinality.

Indeed, as straight line segments from $E$ intersect at most at their endpoints,
segments of $E$ can slightly be shifted to become pairwise disjoint and non-parallel
while keeping all nonempty intersections of subsets of objects from ${\cal{N}}_r(E)$ with some slightly larger $r.$
For two non-overlapping segments $e$ and $e'$ it can be understood that $|{\rm{bd}}\,N_r(e)\cap{\rm{bd}}\,N_r(e')|\leq 2$
because Euclidean distance grows strictly monotonically
from $e$ (or from $e')$ to a point of the curve $\chi(e,e')=\{x\in\mathbb{R}^2:d(x,e)=d(x,e')\}$\footnote{In fact
the curve $\chi(e,e')$ is composed
of pieces of straight lines and parabolas}
as that point moves along $\chi(e,e')$ in any of two opposite directions starting from midpoint of the segment $[x,x']$
which joins the points $x\in e$ and $x'\in e',$ where $d(x,x')=d(e,e')$ and $d(e,e')$ denotes Euclidean distance between
$e$ and $e'.$ Here, obviously, ${\rm{bd}}\,N_r(e)\cap {\rm{bd}}\,N_r(e')\subset\chi(e,e').$
Points from ${\rm{bd}}\,N_r(e)\cap{\rm{bd}}\,N_r(e')$ split ${\rm{bd}}\,N_r(e)$ into two subcurves,
lying on different sides of $\chi(e,e').$ Therefore one of those two subcurves is contained in ${\rm{int}}\,N_r(e'),$
where ${\rm{int}}\,N$ denotes interior of $N\subset\mathbb{R}^2.$
This implies that $N_r(e')\backslash N_r(e)$ is connected and concludes the proof
of the fact that ${\cal{N}}_r(E)$ can be considered to be a set of pseudo-disks.

When pairs of segments of $E$ are allowed to intersect properly and admitted to have
arbitrarily large number of distinct orientations,
it is difficult to achieve a constant factor approximation at least by using known
approaches. It is due to the non-constant lower bound obtained in \cite{alon} on integrality gap of a
geometric intersection problem, which is close to the {\sc IPGD} problem for $r=0.$

\subsection{Our algorithmic contribution}

In contrast to related work this paper is focused on design of
$O(1)$-approxi\-ma\-tion algorithms for the NP-hard {\sc IPGD} problem where segments
from $E$ are allowed to intersect at most at their endpoints, can
have arbitrarily large Euclidean lengths and arbitrarily large number
of distinct orientations. Moreover, the paper focus is on algorithms,
giving a favourable combination of guaranteed constant approximation factor and time
complexity.

There is a major challenge in design of efficient constant factor approximation algorithms for the
{\sc IPGD} problem. Namely, an ugly tradeoff is observed
between guaranteed constant approximation factor and time complexity of the known approximation algorithms
for the problem. This means that in existing approximation algorithms good accuracy (i.e.
approximation factor close to $1)$ can be guaranteed only with high computational cost, which is unacceptable
in practice. Such a situation is very typical in a variety of geometric {\sc Hitting Set} problems on the plane
in which corresponding families ${\cal{R}}$ contain objects of more or less sophisticated shape.

Indeed, exploiting equivalence of the {\sc IPGD} problem to the {\sc PEH} problem
on the set ${\cal{N}}_r(E)$ of pseudo-disks, an $O(1)$-approximation can be designed based on epsilon nets \cite{pyrga}
and Agarwal-Pan iterative reweighting algorithm \cite{agarwal}. It has reasonable time complexity (roughly of $O(n^4))$
but its guaranteed constant approximation factor is large, being more than $100000.$

Of course, PTAS for the {\sc IPGD} problem can be adapted to design an $O(1)$-approximation algorithm
by setting $\varepsilon$ equal to some small constant. According to the introduction in \cite{mulimits},
the resulting constant factor approximation algorithm has huge time complexity of $O(n^{30}).$
Successful attempts are taken to adapt local search for design of faster low constant factor approximation
algorithms. In particular, local search approach has become quite competitive for designing
of low constant factor approximations for
geometric {\sc Hitting Set} problems with families of disks \cite{mulimits} and pseudo-disks \cite{muhall}.
For the {\sc IPGD} problem an algorithm from \cite{muhall}, constructed for pseudo-disks, yields 4-approximate solution
in $O(n^{13})$ time, being too complicated.

Thus, direct application of the aforementioned general algorithms for pseudo-disks to the
{\sc IPGD} problem without relying deeply on
geometry of the range space $(\mathbb{R}^2,{\cal{N}}_r(E))$ results in
the severe tradeoff between their guaranteed constant approximation factor and time complexity.
Some other less general algorithmic techniques fail to work for the {\sc IPGD} problem
due to its relatively general setting.

In this work, to construct approximation algorithms
with better combination of guaranteed approximation factor and time complexity
than that of the aforementioned epsilon net based
approximation algorithm for pseudo-disks, implied by results from \cite{agarwal} and \cite{pyrga},
the {\sc IPGD} problem geometry is incorporated into slightly modified version of the latter algorithm.
Namely, a simple to implement
$\left(50+52\sqrt{\frac{12}{13}}+\nu\right)$-approximation is proposed for the {\sc IPGD} problem
working in $O\left(\left(n^2+\frac{n\log n}{\nu^2}+\frac{\log n}{\nu^3}\right)n^2\log n\right)$ time. Moreover, $(34+24\sqrt{2}+\nu)$- and $\left(\frac{144}{5}+32\sqrt{\frac{3}{5}}+\nu\right)$-appro\-xi\-mations are given for special segment configurations $E$ defined by generalized outerplane graphs and Delaunay triangulations, arising in network applications.
The latter two algorithms work within the same time complexity bounds,
where $\nu>0$ is an arbitrary small constant.
Though guaranteed approximation factors of our algorithms are larger than $4,$
it is very likely that their actual approximation factors are small in practice.

\subsection{Brief description of our algorithms and approaches}\label{gjfjdkksksk}

Constant factor approximation algorithms for the {\sc IPGD} problem proposed in this paper
in fact approximately solve the {\sc PEH}
problem on the corresponding set ${\cal{N}}_r(E)$ of $r$-hippodromes.
Below their general layout and the most important stages are briefly discussed.

\subsubsection{General layout of our approximation algorithms}

Our algorithms contain three main stages.
At the first stage
the {\sc PEH} problem input is preprocessed
to simplify work at the subsequent stages.
The idea of this preprocessing is to remove
from ${\cal{N}}_r(E)$ all $r$-hippodromes,
which have empty intersection with the rest.

{\sc Stage 1: preprocessing.}
Every object $N\in{\cal{N}}_r(E)$ is identified
such that no other object from ${\cal{N}}_r(E)$ intersects $N.$
This can be done in $O(n^2)$ time.
Let ${\cal{K}}$ be a set of all such objects.
One proceeds to the next two stages to find an approximate solution $S$
to the {\sc PEH} problem for the family ${\cal{N}}_r(E)\backslash{\cal{K}}.$
When it is done, a point set $C:=S\cup C_0$ is returned as the final output,
where $C_0=\{c_K:K\in{\cal{K}}\}\subset\mathbb{R}^2$
is such that $c_K\in K$ for each $K\in{\cal{K}}.$

At the second stage the {\sc PEH} problem is discretized. The idea behind this discretization
is as follows. Let $f>0$ be some absolute constant.
Instead of finding an $f$-approximate solution $S$ to the {\sc PEH} problem,
being a finite subset of the whole plane, its special $f$-approximate solution $S_0\subseteq Y_0$ can always be found,
where $Y_0$ is some finite precomputed point set, uniquely defined by the {\sc PEH} problem input.

This discretization allows us to apply a machinery of epsilon nets at the next (third) stage.
This algorithmic machinery is
developed in \cite{agarwal} and \cite{pyrga} for designing of approximation
algorithms for {\sc Hitting Set} problems on discrete range spaces $(Y,{\cal{R}}),$
where discreteness of the range space means that $Y$ is finite.

{\sc Stage 2: discretization.}
Let ${\cal{B}}(E)=\{{\rm{bd}}\,N_r(e):e\in E\}.$ A vertex set of
the arrangement ${\cal{A}}(E,r)$ of curves of ${\cal{B}}(E)$
is used as the set $Y_0.$ Of course, $|Y_0|=O(n^2)$ and $Y_0$
can be constructed in the straightforward way by computing
pairwise intersections of curves from ${\cal{B}}(E).$

It is easy to see that the
{\sc PEH} problem instance can equivalently be reduced to the {\sc Hitting Set}
problem instance for $(Y_0,{\cal{N}}_r(E)).$
Indeed, any
$k$-element hitting set $C$ for ${\cal{N}}_r(E)$ can be converted
into a $k$-element hitting set $C'\subseteq Y_0$ by some sort of (polynomial) point location
algorithm \cite{boiso} on the arrangement ${\cal{A}}(E,r)$
\footnote{Being applied for a point $c\in C$ and the arrangement ${\cal{A}}(E,r),$
point location algorithm identifies
either a vertex, edge or face of the arrangement, which contains $c;$ after that,
either the vertex $c$ or an arbitrary vertex of the edge or face, containing $c,$ can be returned.}.
As a consequence
${\rm{OPT}}(\mathbb{R}^2,{\cal{N}}_r(E))=
{\rm{OPT}}(Y_0,{\cal{N}}_r(E))={\rm{OPT}},$ where
${\rm{OPT}}(Y,{\cal{R}})$
denotes optimum of the {\sc Hitting Set} problem
for a given range space $(Y,{\cal{R}}).$

Then, one
proceeds to the next (main) stage, now dealing with the discrete range space $(Y_0,{\cal{N}}_r(E))$
instead of the range space $(\mathbb{R}^2,{\cal{N}}_r(E)),$
where a point set $Y_0$ is defined as above.
Algorithmic work performed at the main stage is similar in its
structure to the Agarwal-Pan algorithm from \cite{agarwal}.
To describe what is done precisely at this stage,
some notation and definitions are given first.

A map $w:Y_0\rightarrow\mathbb{Q}_+$ is called a {\it weight map} on $Y_0$ in the sequel,
meaning that a positive rational value $w(y)$ is assigned to each point $y\in Y_0.$
Let $w(N)=\sum\limits_{y\in N\cap Y_0}w(y)$ for $N\subseteq\mathbb{R}^2.$
Given ${\cal{N}}\subseteq{\cal{N}}_r(E)$ and a weight map $w$ on $Y_0,$
a triple $(Y_0,{\cal{N}},w)$ is used to denote a range space when
points from $Y_0$ have, generally, non-equal positive rational weights.
Range spaces $(Y_0,{\cal{N}})$ and $(Y_0,{\cal{N}},w_0)$ are considered equivalent, where
$w_0$ is the unit weight map, i.e. $w_0(y)=1$ for all $y\in Y_0$ and $w_0(N)=|N\cap Y_0|$ for $N\subseteq\mathbb{R}^2.$
\begin{definition}
Assume that $0<\varepsilon<1,$ $w:Y_0\rightarrow\mathbb{Q}_+$ and ${\cal{N}}\subseteq{\cal{N}}_r(E)$ are given. Let ${\cal{N}}_{\varepsilon}=\{N\in{\cal{N}}:w(N)>\varepsilon w(Y_0)\}.$
A finite subset $Y'\subseteq\mathbb{R}^2$ is called a {\it (weighted) weak $\varepsilon$-net}
(see also definition in \cite{alon})
for a range space $(Y_0,{\cal{N}},w)$
if $Y'\cap N\neq\varnothing$ for any $N\in{\cal{N}}_{\varepsilon},$
i.e. $Y'$ is a hitting set for ${\cal{N}}_{\varepsilon}.$
\end{definition}

{\sc Main stage: computing weight maps and epsilon nets.}
Work of our algorithms is quite involved at the main stage.
Its most important steps are described briefly at first; then,
a simplified pseudo-code is presented for clarification.

At the main stage the following two steps are performed in turn within a binary search loop:

{\sc Step 1: finding a weight map.}
During the first step either unit weights are assigned to all points
from $Y_0$ or a special
weight map $w:Y_0\rightarrow\mathbb{Q}_+$ is computed.

{\sc Step 2: constructing an epsilon net.}
The second step represents calls of a special procedure,
which, given a weight map $w$ on $Y_0,$
constructs weak $\varepsilon$-nets of cardinality at most $O\left(\frac{1}{\varepsilon}\right)$
for subspaces of the
range space $(Y_0,{\cal{N}}_r(E),w).$
More precisely, let ${\cal{N}}\subseteq{\cal{N}}_r(E),$ a weight map $w:Y_0\rightarrow\mathbb{Q}_+$ and an arbitrary $\varepsilon>0$ are given.
A procedure is referred to as an {\it epsilon net finder} if it
seeks a weak $\varepsilon$-net $C=C(Y_0,{\cal{N}},w,\varepsilon)\subset \mathbb{R}^2$ of size at most $\frac{M}{\varepsilon}$ for $(Y_0,{\cal{N}},w),$ where $M\geq 1$ is an absolute constant, which is specific to this procedure; it is called its performance parameter.
For $\varepsilon\geq 1$ the procedure returns $C=\varnothing.$

Finally, $\varepsilon$ represents the parameter, which is adjusted within the binary
search loop over the two steps above. Namely, using a sort of a trial and error method (see, e.g. \cite{agarwal},\cite{bronnimann}),
its reciprocal $1/\varepsilon$ is adjusted
to be as close to ${\rm{OPT}}$ as possible.
Besides, roughly speaking, at the first step one tries to compute a special weight map $w$ on $Y_0$ such that
${\cal{N}}_{\varepsilon}={\cal{N}}_r(E).$ At the same time,
a chosen epsilon net finder, called at the second step, returns hitting sets for
${\cal{N}}_{\varepsilon}$ of size at most $\frac{M}{\varepsilon}.$ As
$\frac{1}{\varepsilon}$ tends to be close to ${\rm{OPT}},$ one finally gets an $O(1)$-approximate
solution to the {\sc PEH} problem.

Pseudo-code of our algorithms is given below.
Let $\mu_0>0$ and $\mu>1$ be some absolute constants to be defined later.

\smallskip

\hrule
\smallskip
\noindent {\sc Piercing Hippodromes.}
\smallskip
\hrule
\smallskip

\noindent {\bf Input:} $r>0$ and a family ${\cal{N}}_r(E)$ of $r$-hippodromes, where $E$ is an edge set of a plane graph;

\noindent {\bf Output:} $O(1)$-approximate solution to the {\sc PEH} problem for ${\cal{N}}_r(E).$
\smallskip

\hrule
\begin{enumerate}
\item compute ${\cal{K}}=\{K\in{\cal{N}}_r(E):K\cap N=\varnothing\,\,\forall N\in{\cal{N}}_r(E), N\neq K\},$
set $E:=E\backslash\{e\in E:N_r(e)\in {\cal{K}}\},$
$k:=1$ and construct a vertex set $Y_0$ of the arrangement of
curves from ${\cal{B}}(E)=\{{\rm{bd}}\,N_r(e):e\in E\};$

\item find a weak $\frac{1}{\mu_0 k}$-net $C_k$ for $(Y_0,{\cal{N}}_r(E),w_0)$ of size at most $M\mu_0 k;$

\item set ${\cal{N}}_k:=\{N\in{\cal{N}}_r(E):N\cap C_k=\varnothing\};$

\item set $\varepsilon_k:=\frac{1}{\mu k},$ compute a weight map $w_k=w_k(\cdot |Y_0,{\cal{N}}_k,k)$ on $Y_0$ such that $w_k(N)>\varepsilon_kw_k(Y_0)$ for all $N\in{\cal{N}}_k$ if $w_k$ can be algorithmically computed\footnote{At this step the algorithm also determines if such a weight map $w_k$ can in principle be computed.};
    if it can be, set flag:=true; otherwise, set flag:=false;

\item if flag=false, set $k:=2k$ and repeat steps 2-4; otherwise, set $k_p:=k;$

\item repeating steps 2-4 within a binary search loop for $k$ in the interval $(k_p/2, k_p],$ find the smallest $k=k_f\in (k_p/2,k_p]$ for which performing those steps gives flag=true;

\item find a weak $\varepsilon_{k_f}$-net $C'_{k_f}$ for $(Y_0,{\cal{N}}_{k_f},w_{k_f})$ of size at most $\frac{M}{\varepsilon_{k_f}};$

\item return $C=C_0\cup C_{k_f}\cup C'_{k_f}$ as an $M(\mu_0+\mu)$-approximate solution to the {\sc PEH} problem,
where $C_0=\{c_K:K\in{\cal{K}}\}\subset\mathbb{R}^2$ and $c_K\in K$ for each $K\in{\cal{K}}.$
\end{enumerate}
\hrule
\smallskip

Although, pseudo-code of our algorithms tackles the case of nonempty ${\cal{K}},$
it is assumed below that ${\cal{K}}=\varnothing$ for simplicity of presentation.

An important step of the {\sc Piercing Hippodromes} algorithm
is its step 4 where a special weight map $w_k:Y_0\rightarrow\mathbb{Q}_+$ is computed.
A procedure is used to implement this step, which is
a slightly modified version of the iterative reweighting procedure
from \cite{agarwal}. This modified procedure is described
in detail in the section \ref{pkgkslsls}. Its performance analysis, given in the proof of
the theorem \ref{ksakldkekeksa},
shows that the
{\sc Piercing Hippodromes} algorithm finally arrives at the case flag=true
at its step 4. Moreover, the value of $k_p$ is obtained at its step 5 such that
either $k_p/2<{\rm{OPT}}(Y_0,{\cal{N}}_{k_p})\leq k_p$ or $k_p<{\rm{OPT}}(Y_0,{\cal{N}}_{k_p}),$
thus, giving $k_f\leq {\rm{OPT}}(Y_0,{\cal{N}}_{k_f}).$

It is also shown in this paper that parameters $\mu$ and $\mu_0$ can be chosen such that
the upper bound $M(\mu_0+\mu)$ on constant approximation factor of the {\sc Piercing Hippodromes} algorithm
depends mostly on the performance parameter $M$ of the epsilon net finder, running at its steps 2 and 7.
The {\sc Piercing Hippodromes} algorithm time complexity turns out to be of the same order (up to logarithmic factors)
as time complexity of the epsilon net finder.

\subsubsection{Our geometric approaches to design an epsilon net finder}

In this work a special algorithmic scheme is used to devise epsilon net finders, which is based
on ideas from work \cite{pyrga}. Key approaches (given
in \cite{pyrga} and ours) to build such procedures with small
values of $M$ for subspaces of $(Y_0,{\cal{N}}_r(E),w)$ are reported in short below.
Implementation of those procedures and summarizing on their
performances are postponed for the section \ref{fhskakjwwoqo}.

Our epsilon net finders
follow a slightly modified general algorithmic scheme from \cite{pyrga},
applied for subspaces of
the specific range space $(Y_0,{\cal{N}}_r(E),w).$
To give a short description of how those procedures work, let us begin with the following
simple algorithmic idea, which gives an $O(1)$-approximation algorithm
for a special case of the {\sc PEH} problem
in which all segments from $E$ have zero lengths, i.e. the corresponding set ${\cal{N}}_r(E)$ is composed of radius $r$ disks.
Namely, the idea consists in applying a ``divide-and-conquer'' heuristic, which extracts a
{\it maximal independent} set ${\cal{I}}\subseteq{\cal{N}}_r(E)$ of radius $r$
disks within ${\cal{N}}_r(E),$ i.e. a
maximal (with respect to inclusion) subset ${\cal{I}}$ of pairwise non-overlapping disks
from ${\cal{N}}_r(E)$\footnote{In distinction to the known NP-hard Maximum Independent Set problem for disks, the problem of finding
a maximal (by inclusion) subset of non-intersecting disks among disks from ${\cal{N}}_r(E)$ is polynomially solvable. The corresponding algorithm to solve the problem incrementally adds disks into a growing independent set starting from an empty one.}.
It should be noted that
${\cal{N}}_r(E)=\bigcup\limits_{I\in{\cal{I}}}{\cal{N}}_I,$ where
${\cal{N}}_I=\{N\in {\cal{N}}_r(E):N\cap I\neq\varnothing\}$ for $I\in{\cal{I}}.$
As $7$ radius $r$ disks are sufficient
to cover any $2r$ radius disk, for each $I\in{\cal{I}}$ a
$7$-point set $S_I$ can be easily constructed in $O(1)$ time, which has nonempty intersection with
each disk from ${\cal{N}}_I.$
Therefore a set $\bigcup\limits_{I\in{\cal{I}}}S_I$
gives a $7$-approximate solution to the {\sc PEH} problem as $|{\cal{I}}|\leq{\rm{OPT}}.$

In the general setting of the {\sc PEH} problem this ``divide-and-conquer'' heuristic is not working
as one can not guarantee constant sized hitting
sets for ${\cal{N}}_I$ to exist uniformly for all
$I\in{\cal{I}}.$ In fact, it is not possible because, first, Euclidean lengths
are not assumed to be uniformly bounded from above of segments from $E$ by any linear function
of $r.$ Second, one can not cluster segments into constant number of
groups with similar segment orientations and apply this heuristic (with an additional modification)
in each group separately as done e.g. in \cite{nandy} for the case of sets of axis-parallel straight line segments.

Instead, in this paper a similar algorithmic
idea is adopted to design an epsilon net finder.
This idea relies conceptually on the approach of work \cite{pyrga},
being its slightly improved version. Given ${\cal{N}}\subseteq{\cal{N}}_r(E),$
it exercises a similar ``divide-and-conquer'' heuristic,
which suggests to find a maximal subset ${\cal{I}}$
of {\it almost} non-overlapping objects within the set
${\cal{N}}_{\varepsilon}.$
It means that pairs of objects from
${\cal{I}}$ are allowed to have nonempty intersection, but
``amount'' of this intersection does not exceed some fraction of $w(Y_0).$
More precisely, the following definition describes
properties of ${\cal{I}}$ in detail (see also general definition in \cite{pyrga}).
\begin{definition}
Given a subset ${\cal{N}}\subseteq{\cal{N}}_r(E),$ a parameter $0\leq\delta<1$ and a weight map $w:Y_0\rightarrow\mathbb{Q}_+,$
a subset ${\cal{I}}={\cal{I}}(\delta)\subseteq{\cal{N}}$ is called
a {\it maximal (with respect to inclusion) $\delta$-independent} for a range space $(Y_0,{\cal{N}},w)$
if $$w(I\cap I')\leq\delta w(Y_0)$$ for any distinct $I,I'\in{\cal{I}}$
and for any $N\in{\cal{N}}$ there is some $I=I(N)\in{\cal{I}}$ such that $w(N\cap I)>\delta w(Y_0).$
\end{definition}

As $Y_0$ is the vertex set of the arrangement of curves from ${\cal{B}}(E)$
and $w(y)>0$ for every $y\in Y_0,$ any
maximal $0$-independent set for $(Y_0,{\cal{N}},w)$ is
a maximal independent set within ${\cal{N}}.$

Let $0\leq\theta_0<1$ be some parameter to be defined later.
Pseudo-code of our epsilon net finder procedure is given below:

\smallskip

\hrule
\smallskip
\noindent {\sc Weak Epsilon Net Finder.}
\smallskip
\hrule
\smallskip

\noindent {\bf Input:} a range space $(Y_0,{\cal{N}},w)$ for some ${\cal{N}}\subseteq{\cal{N}}_r(E),$
$w:Y_0\rightarrow\mathbb{Q}_+$ and a parameter $0<\varepsilon<1;$

\noindent {\bf Output:} a weak $\varepsilon$-net for $(Y_0,{\cal{N}},w)$
of size at most $\frac{M}{\varepsilon}$ for some absolute constant $M\geq 1.$
\smallskip

\hrule

\begin{enumerate}

\item set $\delta:=\theta_0\varepsilon,$ find a maximal $\delta$-independent set ${\cal{I}}={\cal{I}}(\delta)\subseteq
{\cal{N}}_{\varepsilon}$ for $(Y_0,{\cal{N}}_{\varepsilon},w)$ and form disjoint sets ${\cal{N}}_{\delta,I},\,I\in {\cal{I}},$
    such that $\bigcup\limits_{I\in{\cal{I}}}{\cal{N}}_{\delta,I}={\cal{N}}_{\varepsilon},$
    where  ${\cal{N}}_{\delta,I}\subseteq{\cal{N}}^0_{\delta,I}=\{N\in{\cal{N}}_{\varepsilon}:w(N\cap I)>\delta w(Y_0)\};$

\item for each $I\in{\cal{I}}$ compute a hitting set
$C_I$ for ${\cal{N}}_{\delta,I}$ of size at most $\frac{c_1w(I)}{\delta w(Y_0)}+c_2$
for some positive constants $c_1$ and $c_2;$

\item return the set $C_{\theta_0}=\bigcup\limits_{I\in{\cal{I}}}C_I.$
\end{enumerate}
\hrule
\smallskip

At its step 1 the {\sc Weak Epsilon Net Finder} procedure forms a partition of ${\cal{N}}_{\varepsilon}$ into
disjoint subsets ${\cal{N}}_{\delta,I};$ here each subset ${\cal{N}}_{\delta,I}$ does not have to
coincide with ${\cal{N}}^0_{\delta,I},$ being a result of the partitioning process for ${\cal{N}}_{\varepsilon}.$
The way in which the partitioning is done within the {\sc Weak Epsilon Net Finder} procedure is slightly different from the way in which it is done in an epsilon net finder, resulting from applying the approach of work \cite{pyrga}
directly to subspaces of $(Y_0,{\cal{N}}_r(E),w).$
It is demonstrated in the subsubsection \ref{skskfkgeoqoapdlgls} that our modification gives the smaller value of the
performance parameter $M.$
As $|Y_0|=O(n^2)$ the procedure step 1
can be implemented in a straightforward way without relying on the {\sc PEH} problem specifics
(see proof of the lemma \ref{dskskskkroroeowoqp}
from the subsection \ref{gkfkkskslf_glro}).

There are two basic parts of the {\sc Weak Epsilon Net Finder} procedure implementation
which rely on geometry of shape of $r$-hippodromes
and the {\sc PEH} problem specifics. They are the main algorithmic result of this paper.
The first part consists in a special subprocedure
which is aimed at computing a hitting set of size at most $\frac{c_1w(I)}{\delta w(Y_0)}+c_2$
with small nonnegative constants $c_1$ and $c_2$ for a family ${\cal{N}}_{\delta,I},$
corresponding to an object $I\in{\cal{I}}$
of any maximal $\delta$-independent set ${\cal{I}}$ for $(Y_0,{\cal{N}}_{\varepsilon},w).$
In the subsection \ref{qeiaididis} several fast variants are given of such subprocedures,
guaranteeing small constants $c_1$ and $c_2$ (also called their performance parameters)
under different assumptions on $E,$ including the general case in which $E$ belongs to a class ${\cal{E}}_0$
 of edge sets of arbitrary plane graphs. Those subprocedures in fact construct weak $\Delta_I$-nets for subspaces $(Y_0\cap I,{\cal{N}},w_I),\,I\in{\cal{I}},$
of size at most $\frac{c_1}{\Delta_I}+c_2,$ where $\Delta_I=\frac{\delta w(Y_0)}{w(I)}$
and $w_I=w\mid_{Y_0\cap I}.$ They are called {\it subspace epsilon net finders} in the sequel.

The second part consists in choosing a suitable value of the parameter $\theta_0,$ scaling
the upper bound $\delta w(Y_0)$ on intersection weight of pairs of objects from ${\cal{I}}.$
Indeed, it is not an easy task to choose $\theta_0$ to get a modest value
of the performance parameter $M.$
Namely, there is a tradeoff between constants in the two bounds
$|{\cal{I}}|=O\left(\frac{1}{\varepsilon}\right)$ and $\frac{\sum\limits_{I\in{\cal{I}}}w(I)}{\delta w(Y_0)}=
O\left(\frac{1}{\varepsilon}\right),$ implying the bound $|C_{\theta_0}|\leq\frac{M}{\varepsilon}.$
Setting $\theta_0$ either large or close to zero may result in a large
constant either in the bound $|{\cal{I}}|=O\left(\frac{1}{\varepsilon}\right)$ or in the bound
$\frac{\sum\limits_{I\in{\cal{I}}}w(I)}{\delta w(Y_0)}=
O\left(\frac{1}{\varepsilon}\right)$ respectively.

In the subsection \ref{gkfkkskslf_glro} a geometric approach is used to adjust the parameter $\theta_0,$ originating from work \cite{pyrga}.
Given a subclass ${\cal{E}}\subseteq{\cal{E}}_0,$
the approach allows to get a value $\theta^{\ast}_0=\theta^{\ast}_0({\cal{E}})$ of this parameter, finally leading to a small
value $M^{\ast}=M^{\ast}({\cal{E}})$ of $M$ in the upper bound on $|C_{\theta^{\ast}_0}|,$ which
holds true uniformly within the class ${\cal{E}}.$ As a result, it guarantees a small
constant upper bound on the approximation factor of the {\sc Piercing Hippodromes} algorithm to hold within ${\cal{E}}.$

For the case $c_1=0$ minimizing the performance parameter $M$ with respect to $\theta_0$
implies that $\theta_0\rightarrow 0.$
Thus, in this case the {\sc Weak Epsilon Net Finder} procedure is similar in its layout to the 7-approximation algorithm, described above for the case where $E$ consists of points.

\section{Weak Epsilon Net Finder: implementation and performance analysis}\label{fhskakjwwoqo}

\subsection{Subspace epsilon net finders}\label{qeiaididis}

Three subquadratic subspace epsilon net finders are
constructed below to be applied at the step 2 of the {\sc Weak Epsilon Net Finder} procedure.
These procedures have small performance parameters $c_1$ and $c_2,$ being designed for
different classes of sets of non-zero length straight line segments.
They largely amount to fast construction of the least cardinality hitting sets for sets of 1-dimensional intervals on the real line.

More specifically, the first subspace epsilon net finder treats the general case where
segments from $E$ form an edge set of some plane graph $G,$ giving $c_1=8$ and $c_2=2.$
The second subspace epsilon net finder works for the special case where $E$ is such that
either $d(e,e')>r$ or $d(e,e')=0$ for any distinct $e,e'\in E.$ It reports $c_1=1$ and $c_2=6.$
The last subspace epsilon net finder tackles the case in which $E$ is being an
edge set of any subgraph of a Delaunay triangulation $G.$
Graphs of this type arise in network routing and modelling applications. This procedure gives $c_1=4$ and $c_2=8.$

Let a subset ${\cal{N}}\subseteq{\cal{N}}_r(E),$ an object $I\in{\cal{N}},$
an arbitrary $0<\Delta<1,$ a nonempty finite point set $F\subset\mathbb{R}^2$
and a weight map $w_I:F\cap I\rightarrow\mathbb{Q}_+$ are given.
Generally, proposed subspace epsilon net finders are assumed to
accept an object $I$ and a subset
${\cal{N}}_I(\Delta)\subseteq\{N\in{\cal{N}}:w_I(N)>\Delta w_I(I)\}$
as their input. Their output is a hitting set for ${\cal{N}}_I(\Delta)$ of size at most $\frac{c_1}{\Delta}+c_2$ for some nonnegative constants
$c_1$ and $c_2.$ It is shown below that they work in $O(m\log m)$ time and $O(m)$ space, where $m=|{\cal{N}}_I(\Delta)|.$

When they are applied within the {\sc Weak Epsilon Net Finder} procedure,
their input is an object $I$ of a maximal $\delta$-independent set for $(Y_0,{\cal{N}}_{\varepsilon},w)$ and a subset ${\cal{N}}_{\delta,I}$ from
the corresponding partition of ${\cal{N}}_{\varepsilon},$ built at its step 1.
It corresponds to setting $F=Y_0,$ $\Delta=\frac{\delta w(Y_0)}{w(I)}$ and $w_I=w\mid_{Y_0\cap I}.$

\subsubsection{General case of plane $G$}\label{qyuuififis}

Let $E$ be an edge set of a plane graph $G.$
The following observations can be made
about shape of $r$-hippodromes.
\medskip

\noindent {\bf Observation 1.} Let $e,e'\in E$ be such that $M=N_r(e)\cap N_r(e')\neq\varnothing.$
Then $M=N_r(z_{e'}(e))\cap N_r(e')=N_r(e)\cap N_r(z_e(e')),$ where $z_e(e')=\{x\in e':d(x,e)\leq 2r\}.$
\medskip

Let $l(e)$ be a straight line through $e$ for some non-zero length segment $e\in E$
and $h_1(e)$ and $h_2(e)$ be positive and negative halfplanes respectively
whose boundary coincides with $l(e);$ here
orientation is chosen arbitrarily for $l(e).$ The set ${\rm{bd}}\,N_r(e)$ can
be represented in the form of a union of two halfcircles and
two segments $f_1(e)$ and $f_2(e),$ where $f_i(e)\subset {\rm{int}}\,h_i(e),\,i=1,2.$
Let $l_i(e)$ be the straight line through $f_i(e).$

\medskip
\noindent {\bf Observation 2.} Let $\{v_1,v_2\}=l(e)\cap{\rm{bd}}\,N_r(e),\,e\in E.$
For every $e,e'\in E$ for which $N_r(e)\cap N_r(e')\neq\varnothing$ either
$N_r(e')\cap\{v_1,v_2\}\neq\varnothing$ or
such $i_0\in\{1,2\}$ exists that $d(x,l_{i_0}(e))\leq r$ for all $x\in z_{e}(e').$
\medskip

Let us provide some notation. Given a subset
${\cal{M}}\subseteq{\cal{N}}_r(E),$ a subset $E({\cal{M}})$ is such that
${\cal{M}}={\cal{N}}_r(E({\cal{M}}));$ in particular, for $N\in{\cal{N}}_r(E)$
let $e(N)\in E$ be a segment such that $N=N_r(e(N)).$

Our subspace epsilon net finder is based on finding hitting sets for sets of
1-dimensional $r$-neighbourhoods of (interval) projections of segments from
$\{z_e(e')\}_{e'\in E'},\,E'\subseteq E,$ onto straight lines $l_i(e).$
Let $N_{ir}(f)=\{x\in l_i(e):d(x,f)\leq r\}$ for
an arbitrary interval $f\subset l_i(e),\,i=1,2.$
The following folklore lemma reports on the complexity
of getting minimum cardinality hitting set for a set of 1-dimensional intervals.
Its proof is left for the appendix.
\begin{lemma}\label{gjfkdkksaadqqrsra}
The minimum cardinality hitting set can be found for a set of $n$ 1-dimensional intervals
on the real line in $O(n\log n)$ time and $O(n)$ space.
\end{lemma}

\smallskip

\hrule
\smallskip
\noindent {\sc Subspace Weak Epsilon Net Finder }
\smallskip
\hrule
\smallskip

\noindent {\bf Input:} an object $I\in {\cal{N}}$ and a set ${\cal{N}}_I(\Delta);$

\noindent {\bf Output:} a hitting set $C_I\subset\mathbb{R}^2$ for ${\cal{N}}_I(\Delta).$
\smallskip

\hrule

\begin{enumerate}

\item set $\{v_1,v_2\}=l(e(I))\cap{\rm{bd}}\,I$ and
$${\cal{P}}:={\cal{N}}_I(\Delta)\backslash\{N\in{\cal{N}}:N\cap\{v_1,v_2\}\neq\varnothing\};$$

\item form sets $Z_i=\{z_{e(I)}(e):e\in E({\cal{P}}),\,z_{e(I)}(e)\subset h_i(e(I))\},\,i=1,2;$

\item form a set $P_i$ of orthogonal projections of segments from $Z_i$
onto the straight line $l_i(e(I))$ and construct sets $P_i(r)=\{N_{ir}(p):p\in P_i\},\,i=1,2;$

\item find the minimum cardinality hitting set
$H_i\subset l_i(e(I))$ for $P_i(r),\,i=1,2,$ as in the proof of the lemma \ref{gjfkdkksaadqqrsra};

\item for each $x_0\in H_i$ and $i=1,2$ construct a set $S(x_0)$ of 4 points such that
$N_{\sqrt{2}r}(x_0)\subset\bigcup\limits_{x\in S(x_0)}N_r(x)$ and return
a set $C_I=\{v_1,v_2\}\cup\bigcup\limits_{x_0\in H_i,i=1,2}S(x_0).$

\end{enumerate}
\hrule
\smallskip

The following lemma summarizes on the procedure performance.
\begin{lemma}\label{kob11111111}
Let $m=|{\cal{N}}_I(\Delta)|.$
The {\sc Subspace Weak Epsilon Net Finder} procedure returns
a hitting set $C_I$ for ${\cal{N}}_I(\Delta)$
of size at most $\frac{8}{\Delta}+2$ in $O(m\log m)$ time
and $O(m)$ space.
\end{lemma}
\begin{proof}
All steps except for step 4 of the procedure require $O(m)$ time whereas step 4 takes $O(m\log m)$ time
according to the lemma \ref{gjfkdkksaadqqrsra}. It remains to get the bound $|C_I|\leq \frac{8}{\Delta}+2$
and prove that $C_I$ is a hitting set for ${\cal{N}}_I(\Delta).$
Indeed, due to the step 1 and the observation 2 one has either $z_{e(I)}(e)\subset h_1(e(I))$
or $z_{e(I)}(e)\subset h_2(e(I))$ for every $e\in E({\cal{P}}).$
Moreover, each interval $J\in P_i(r)$ is an orthogonal projection of some object
$P^{-1}_i(J)\in {\cal{N}}_r(Z_i).$ According to the proof of the lemma \ref{gjfkdkksaadqqrsra}, for each
$i=1,2$ at the step 4
a maximal subset $Q_i\subseteq P_i(r)$ is built
of pairwise non-overlapping intervals with $|Q_i|=|H_i|.$ Thus, the respective set
$\{P^{-1}_i(J):J\in Q_i\}$ consists of non-intersecting objects. By the observation 1
one gets that $w_I(P^{-1}_i(J)\cap I)>\Delta w_I(I)$ for all $J\in Q_i.$ Therefore
$|Q_i|\leq \frac{1}{\Delta}$ and $|C_I|=4|Q_1|+4|Q_2|+2\leq\frac{8}{\Delta}+2.$

By the observation 2 each point of a segment from $Z_i$ is within the distance $r$ from $l_i(e(I)).$
Therefore each segment of $Z_i$ is within $\sqrt{2}r$ distance from some point of $H_i.$
By construction at the step 5 one gets that $C_I$ is a hitting set for ${\cal{N}}_I(\Delta).$
\end{proof}

\subsubsection{Case where edges of $G$ are far apart}\label{hkfksiwiaoaofohdov}

Below a special case is considered in which
either $d(e,e')>r$ or $d(e,e')=0$ for any distinct $e,e'\in E.$
For this case
the following idea can be used to implement a subspace epsilon net finder
with smaller $c_1$ and $c_2.$
The idea consists in exploiting the fact that a small constant sized point set $U(I)\subset\mathbb{R}^2$ can be computed in $O(1)$ time for which the set
${\cal{P}}=\{N\in {\cal{N}}_I(\Delta):N\cap U(I)=\varnothing\}$ can be transformed into the set
${\cal{J}}_I(\Delta)=\{{\rm{bd}}\,I\cap N: N\in {\cal{P}}\}$ of 1-dimensional arcs, where the following property $(\ast)$
holds for ${\cal{J}}_I(\Delta):$
if ${\cal{M}}\subseteq{\cal{P}}$ is such that
$I\cap\bigcap\limits_{N\in {\cal{M}}}N\neq\varnothing,$
then ${\rm{bd}}\,I\cap \bigcap\limits_{N\in {\cal{M}}}N\neq\varnothing.$
More precisely, the idea suggests to exclude from ${\cal{N}}_I(\Delta)$ those objects which are hit by $U(I)$ and reduce the problem
of computing a small hitting set for the remaining objects to the equivalent much
simpler problem of finding a hitting set for the corresponding set of one-dimensional arcs.

For an object $I\in{\cal{N}}$
let $C(I)$ be the set of 4 endpoints of segments $f_i(e(I))$
and $U(I)=C(I)\cup (l(e(I))\cap {\rm{bd}}\,I),$ where $i=1,2$ and $|U(I)|=6.$
As a start, a simple observation can be made
about shape of $r$-hippodromes
of non-zero length segments.
\medskip

\begin{lemma}\label{ahhhahshsfafadsdwrwr}
Let $I,N_1,N_2\in{\cal{N}}$ be distinct
and $d(e(I),e(N_i))\in (r,2r],\,i=1,2.$
If $I\cap N_1\cap N_2\neq\varnothing,$ then either
$N_1\cap N_2\cap {\rm{bd}}\,I\neq\varnothing$ or $N_{i_0}\cap U(I)\neq\varnothing$
for some $i_0\in\{1,2\}.$
\end{lemma}
\begin{proof}
Let $\chi_i={\rm{bd}}\,I\cap N_i$ and
$\pi_i={\rm{bd}}\,N_i\cap I$ for $i=1,2.$
Assume that $N_i\cap U(I)=\varnothing$ for all $i\in\{1,2\}.$ It should be
proved that $\chi_1\cap \chi_2\neq\varnothing$ if $I\cap N_1\cap N_2\neq\varnothing.$
Let $p(x)\in e(I)$ be Euclidean projection of $x$ onto $e(I)$
for $x\in\mathbb{R}^2.$ It is sufficient to establish
the following monotonicity property: for any $x\in {\rm{bd}}\,I$ and $i=1,2$
nonempty intersection $[p(x),x]\cap N_i$ is a (possibly zero-length) segment with its endpoint in $x.$
Indeed, for $x\in I\cap N_1\cap N_2$ this implies that the ray
with the origin $p(x)$ and direction $x-p(x)$ intersects ${\rm{bd}}\,I$
at some point of $\chi_1\cap \chi_2.$

Suppose, in contrary, there is a point $x_0\in {\rm{bd}}\,I$
and $i_0\in\{1,2\}$ such that
the interval $(p(x_0),x_0)$ has two (possibly identical) points $x'_1$ and $x'_2$
of intersection with $\pi_{i_0}.$ There is a point $x'\in [p(x_0),x_0]$
and a endpoint $x''\in e(N_{i_0})$ with $(x'-x'',x_0-p(x_0))=0,$
\footnote{$(\cdot,\cdot)$ denotes Euclidean scalar product in $\mathbb{R}^2.$}
such that $d(x',x'')\leq r.$
It implies the inclusion $x''\in \bigcup\limits_{x\in U(I)}N_r(x)$ taking
into account that $r<d(e(I),e(N_{i_0}))\leq 2r.$ But this inclusion is impossible
by our assumption that $N_{i_0}\cap U(I)=\varnothing.$
\end{proof}

By construction of the set ${\cal{P}}$ there is a point in ${\rm{bd}}\,I,$ which does not belong to $\bigcup\limits_{J\in{\cal{J}}_I(\Delta)}J.$
Therefore the property $(\ast)$ holds true  for ${\cal{P}}$ and ${\cal{J}}_I(\Delta)$ by Helly theorem
and the lemma \ref{ahhhahshsfafadsdwrwr}. Thus, the problem of finding the smallest cardinality
hitting set for ${\cal{P}}$ is equivalent to the problem of finding the least cardinality hitting set
for ${\cal{J}}_I(\Delta).$ The latter problem can be equivalently reduced to the problem of computing the smallest size hitting set for a set of one-dimensional intervals on the real line, using polar coordinates.

Our subspace epsilon net finder is given below. It amounts to constructing
a minimum cardinality hitting
set for the set of ``1-dimensional'' intervals ${\cal{J}}_I(\Delta).$

\smallskip

\hrule
\smallskip
\noindent {\sc Subspace Weak Epsilon Net Finder$^{\ast}$}
\smallskip
\hrule
\smallskip

\noindent {\bf Input:} an object $I\in {\cal{N}}$ and a set ${\cal{N}}_I(\Delta);$

\noindent {\bf Output:} a hitting set $C_I\subset\mathbb{R}^2$ for ${\cal{N}}_I(\Delta).$
\smallskip

\hrule

\begin{enumerate}

\item compute $U(I)$ as described before the lemma \ref{ahhhahshsfafadsdwrwr};

\item set ${\cal{P}}:=\{N\in{\cal{N}}_I(\Delta):N\cap U(I)=\varnothing\}$
and ${\cal{J}}_I(\Delta):=\{{\rm{bd}}\,I\cap N:N\in {\cal{P}}\};$

\item applying polar coordinates, find the minimum cardinality hitting set
$C'_I$ for ${\cal{J}}_I(\Delta)$ as in the proof of the lemma \ref{gjfkdkksaadqqrsra};

\item return $C_I=C'_I\cup U(I).$

\end{enumerate}
\hrule
\smallskip

Based on the lemma \ref{ahhhahshsfafadsdwrwr}
performance analysis is given below of
the {\sc Subspace Weak Epsilon Net Finder$^{\ast}$} procedure
under a weaker assumption on the set $E.$

\begin{lemma}\label{kkfkfkdkdksks}
Let $m=|{\cal{N}}_I(\Delta)|$ and ${\cal{P}}={\cal{N}}_I(\Delta)\backslash\{N\in{\cal{N}}:N\cap U(I)\neq\varnothing\}.$
If $E({\cal{P}})$ consists of segments at the distance more than $r$ from $e(I),$
then the {\sc Subspace Weak Epsilon Net Finder$^{\ast}$} procedure returns
a hitting set $C_I$ for ${\cal{N}}_I(\Delta)$
of size at most $\frac{1}{\Delta}+6$ in $O(m\log m)$ time
and $O(m)$ space.
\end{lemma}
\begin{proof}
The set $C_I$ gives a hitting
set for ${\cal{N}}_I(\Delta)$ as $C'_I$ is a hitting set for ${\cal{J}}_I(\Delta)$
by construction reported in the proof of the lemma \ref{gjfkdkksaadqqrsra}.
Thus, it remains to estimate $|C'_I|.$
As byproduct of this construction one gets a maximal
subset ${\cal{J}}'$ of non-overlapping arcs from ${\cal{J}}_I(\Delta)$ with $|C'_I|=|{\cal{J}}'|.$
Let ${\cal{P}}'\subseteq{\cal{P}}$ be the subset such that
${\cal{J}}'=\{{\rm{bd}}\,I\cap N:N\in{\cal{P}}'\}.$
Due to the lemma \ref{ahhhahshsfafadsdwrwr} ${\cal{P}}'$
consists of non-overlapping objects within $I.$
Therefore $|C'_I|\leq\frac{1}{\Delta}.$
\end{proof}

\subsubsection{Case of Delaunay triangulation $G$}

Subspace epsilon net finders are constructed below under the assumption that
$E$ is a special configuration of straight line segments.
It is produced from a planar finite point set $V$ using the so called empty disk property:
two points $u,v\in V$ are joined by a straight line segment when a disk exists which
contains $u$ and $v$ on its boundary and does not contain any other points from $V.$
This property makes the other segments of $E$ avoid having their endpoints in that disk.
It defines a special class of plane graphs called Delaunay triangulations.

\begin{definition}
Assuming that no $4$ points of $V$ are cocircular
a plane graph $G=(V,E)$ is called a {\it Delaunay
triangulation} when $[u,v]\in E$
iff there is a disk $D$ such that
$u,v\in {\rm{bd}}\,D$ and $V\cap {\rm{int}}\,D=\varnothing.$
Such a disk $D$ is called an {\it empty disk} for $[u,v].$
\end{definition}

Let $E$ be an edge set of an arbitrary subgraph of a Delaunay triangulation $G.$
A subspace epsilon net finder is constructed below with small $c_1$ and $c_2.$
Let ${\cal{N}}\subseteq{\cal{N}}_r(E),$ $I\in{\cal{N}}$
and ${\cal{N}}_I(\Delta)$ be given as its input.
By definition of Delaunay triangulation there is a disk $D(e(I))$
such that both endpoints of $e(I)$ lie on its boundary and
none of segments from $E$
has its endpoints in the interior of $D(e(I)).$ Let $c(e(I))$ be its center.

Each segment $e\in E({\cal{N}}_I(\Delta))$ must intersect $N_{2r}(e(I)).$
In the lemma 6 of \cite{kobylkin_dryakhlova} it is shown that for $c(e(I))\in l(e(I))$
an at most $14$-point set $U_0(I)$ can be computed in $O(1)$ time such that
$N\cap U_0(I)\neq\varnothing$ for every $N\in{\cal{N}}_I(\Delta).$
Therefore if $c(e(I))$ is, say, in a halfplane $h_{i_0}(e(I))$ for $i_0\in\{1,2\},$ there must
be a constant sized point set, which hits all objects from ${\cal{N}}_{i_0}(e(I))=\{N\in{\cal{N}}_I(\Delta):e(N)\cap {\rm{cl}}\,h_{i_0}(e(I))\cap N_{2r}(e(I))\neq\varnothing\}.$ More precisely, it can be proved that

\begin{lemma}\label{jh_}
An at most $8$-point set $U(I)$ can be found in $O(1)$ time such that
$N\cap U(I)\neq\varnothing$ for every $N\in{\cal{N}}_{i_0}(e(I)).$
\end{lemma}

Before giving the lemma proof a procedure is described first
which computes the subset $U(I).$ Its steps are as follows:
\medskip

\hrule
\smallskip
\noindent {\sc Constant Hitting Set Finder.}
\smallskip
\hrule
\smallskip

\noindent {\bf Input:} a constant $r>0$ and an edge $e=[u_1,u_2]$ of
a Delaunay triangulation $G=(V,E);$

\noindent {\bf Output:} an at most $8$-point hitting set $U(e)\subset\mathbb{R}^2$
for ${\cal{N}}_{i_0}(e)=\{N\in{\cal{N}}_r(E):e(N)\cap {\rm{cl}}\,h_{i_0}(e)\cap N_{2r}(e)\neq\varnothing\},$ where $c(e)$ is
a center of an empty disk $D(e)$ for $e$ and $c(e)\in {\rm{cl}}\,h_{i_0}(e)$ for some $i_0\in\{1,2\}.$
\smallskip

\hrule
\begin{enumerate}

\item for each $s\in\{1,2\}$ construct a regular hexagon inscribed in $N_{2r}(u_s),$
whose orientation is such that the straight line through $e$ contains a pair of vertices of
that hexagon; form a 7-point set $V_s,$ which contains $u_s$ and midpoints
of the hexagon sides (of length $2r)$;

\item let $v_{s1}$ and $v_{s2}$ be points from $V_s,$ which are symmetric with respect to
$u_s$ and $[v_{s1},v_{s1}]\perp e,\,s=1,2;$

\item for each $s\in\{1,2\}$ choose a subset $U_s\subset V_s$ such that $|U_s|=3$ and
$T_s\cap h_{i_0}(e)\subset\bigcup\limits_{u\in U_s}N_r(u),$ where $N_{2r}(e)=T_1\cup T_2\cup R$
for some rectangle $R$ and two closed halfdisks $T_1$ and $T_2$ of radii $2r$
centered at $u_1$ and $u_2$ respectively; set $u_{si_0}:=g_{i_0}(e)\cap {\rm{bd}}\,N_r(v_{si_0}),$
where $g_{i_0}(e)\subset h_{i_0}(e)\cap{\rm{bd}}\,N_{2r}(e)$ is a straight line segment, touching both
$T_1$ and $T_2;$

\item let $\Delta(e)=\sqrt{d(u_1,u_2)^2/4+d(c(e),e)^2};$

\item if either $\Delta(e)\geq\frac{(2\sqrt{3}-1)r}{\sqrt{4\sqrt{3}-6}}$ or $\Delta(e)\in(0,r/2],$
return $U(e):=U_1\cup U_2;$

\item for $\Delta(e)\in \left[2r,\frac{(2\sqrt{3}-1)r}{\sqrt{4\sqrt{3}-6}}\right)$
set $u_0:=\frac{u_1+u_2}{2}$ and construct two points $z_{1i_0}$ and $z_{2i_0}$
at the intersection $g_{i_0}(e)\cap {\rm{bd}}\,N_{\Delta(e)}(u_0),$ where
$z_{si_0}$ is closer to $u_{si_0}$ than the point $z_{(3-s)i_0}$ is;
set $a_{si_0}=\frac{u_{si_0}+z_{si_0}}{2},\,s=1,2;$

\item if $\Delta(e)\in \left(r,2r\right),$
consider a (rectangular) coordinate system with the origin at
$u_s$ whose $x$-axis is along $e$ and $y$-axis
is perpendicular to $x$-axis, being directed towards $g_{i_0}(e);$
set $b_{i_0}=(d(u_1,u_2)/2,\Delta(e))$ and $a_{si_0}=\frac{u_{si_0}+b_{i_0}}{2},\,s=1,2;$

\item for $\Delta(e)\in (r/2,r]$ set $a_{1i_0}:=(d(u_1,u_2)/2,\sqrt{3}r),$
$a_{2i_0}:=\left(d(u_1,u_2)/2,\frac{\sqrt{3}r}{2}\right);$

\item return
$U(I):=U_1\cup U_2\cup\{a_{1i_0},a_{2i_0}\}.$

\end{enumerate}
\hrule
\medskip

\begin{proof}
Let $U(I)$ be a set, produced by the {\sc Constant Hitting Set Finder}
procedure for $e=e(I).$ Let $O_I=N_{\Delta(e(I))}(u_0)\cap h_{i_0}(e(I))\cap R.$
Obviously, $O_I\subset D(e(I)).$
Using planarity of Delaunay triangulations and applying the same argument as in the proof of the lemma 6
from \cite{kobylkin_dryakhlova} it can be proved that $U(I)$
hits all objects from ${\cal{N}}_{i_0}(e(I)).$
\end{proof}

\begin{lemma}\label{jh}
A hitting set $C'_I$ of size at most $\frac{4}{\Delta}$
can be constructed for ${\cal{M}}_I(\Delta)=\{N\in{\cal{N}}_I(\Delta):N\cap U(I)=\varnothing\},$
using the {\sc Subspace Weak Epsilon Net Finder} procedure, applied for
$I$ and ${\cal{M}}_I(\Delta).$
\end{lemma}
\begin{proof}
As $z_{e(I)}(e(N))\subset h_{3-i_0}(e(I))\cap N_{2r}(e(I))$
for every $N\in{\cal{M}}_I(\Delta),$ the set $C'_I$ is a hitting set
for ${\cal{M}}_I(\Delta)$ of size at most $\frac{4}{\Delta}$ due to
the proof of the lemma \ref{kob11111111}.
\end{proof}

Finally, $C_I=U(I)\cup C'_I$ gives a hitting set for ${\cal{N}}_I(\Delta)$
of size at most $\frac{4}{\Delta}+8.$
The lemma below summarizes on the performance of the described subspace epsilon net finder.
It follows from lemmas \ref{jh_} and \ref{jh}.

\begin{lemma}\label{fsjdkgowoqpspa}
Let $m=|{\cal{N}}_I(\Delta)|.$
A hitting set $C_I$ for ${\cal{N}}_I(\Delta)$ can be built within $O(m\log m)$ time
and $O(m)$ space of size at most $\frac{4}{\Delta}+8.$
\end{lemma}

\subsection{Estimating the parameter $\theta_0$ of the Weak Epsilon Net Finder procedure}\label{gkfkkskslf_glro}

\subsubsection{General approach to estimate $\theta_0$}

Below, given a class ${\cal{E}}\subseteq{\cal{E}}_0$
of sets of straight line segments on the plane,
a general approach is described below to adjust the parameter $\theta_0$
of the {\sc Weak Epsilon Net Finder} procedure. Adjusting $\theta_0$ allows to guarantee
an upper bound $|C_{\theta_0}|\leq\frac{M}{\varepsilon}$ to hold uniformly for all $E\in{\cal{E}},$
where $C_{\theta_0}$ is a weak $\varepsilon$-net, output by the procedure,
and $M$ is some absolute constant, defining the procedure performance parameter for the class ${\cal{E}}.$
Moreover, tuning $\theta_0$ is aimed at minimizing $M.$

More precisely, the approach suggests to identify
a special mapping, defined for any $E\in {\cal{E}}$ and any ${\cal{I}}\subseteq{\cal{N}}_r(E).$
Existence of this special mapping guarantees the bound $|C_{\theta^{\ast}_0}|\leq\frac{M^{\ast}}{\varepsilon}$
to hold uniformly within ${\cal{E}}$ for some $\theta^{\ast}_0=\theta^{\ast}_0({\cal{E}})$
and $M^{\ast}=M^{\ast}({\cal{E}}),$ where both
$M^{\ast}$ and $\theta^{\ast}_0$ depend on constant parameters of the identified
mapping and performance parameters $c_1$ and $c_2$ of a subspace
epsilon net finder, applied at the step 2 of the {\sc Weak Epsilon Net Finder} procedure.
The approach follows the paper \cite{pyrga}.
The lemma below describes how it works.

\begin{lemma}\label{dskskskkroroeowoqp}
Let ${\cal{E}}\subseteq{\cal{E}}_0$ be a class of sets of straight line segments on the plane.
Assume that:
\begin{enumerate}
\item there are absolute constants $\alpha,\beta,\tau$ and
a graph $G_{{\cal{I}}}=({\cal{I}},U)$ for every $E\in{\cal{E}}$ and ${\cal{I}}\subseteq {\cal{N}}_r(E)$
such that $|U|\leq\beta|{\cal{I}}|$ and
$m_{{\cal{I}}}(y)\geq\alpha n_{{\cal{I}}}(y)-\tau$ for every $y\in Y_0,$ where
$n_{{\cal{I}}}(y)=|\{I\in {\cal{I}}:y\in I\}|$ and
$m_{{\cal{I}}}(y)=|\{\{I,I'\}\in U:y\in I\cap I'\}|;$

\item the {\sc Weak Epsilon Net Finder} procedure is applied with a subspace epsilon net finder,
whose time complexity is $O(Q(m)),$ where $m$ is the number of objects in a
subset of ${\cal{N}}_r(E),$ defining an input of the subspace epsilon net finder,
and there is a constant $L>0$ such that $Q(m_1)+\ldots+Q(m_t)\leq L Q\left(\sum\limits_{k=1}^t m_k\right)$ for any positive integers $t, m_1,\ldots m_t.$
\end{enumerate}
Then, given $E\in{\cal{E}}$ and ${\cal{N}}\subseteq{\cal{N}}_r(E),$
a weak $\varepsilon$-net is constructed for $(Y_0,{\cal{N}},w)$ by the {\sc Weak Epsilon Net Finder} procedure
for any $0<\varepsilon<1$
of size at most $$\left[\left(1+\frac{1}{\sqrt{1+\frac{c_2\alpha}{c_1\beta}}}\right)\left(\frac{2c_1\tau\beta}{\alpha^2}+\frac{c_2\tau}{\alpha}\right)+\frac{c_2\tau}{\alpha \sqrt{1+\frac{c_2\alpha}{c_1\beta}}}\right]\frac{1}{\varepsilon}$$ in
$O\left(\frac{\tau n^3}{\alpha\varepsilon}+Q(n)\right)$ time and linear space with respect to the space used to store $(Y_0,{\cal{N}},w),$ where
\begin{equation}\label{hjhiislsldlsls}
\theta_0=\theta^{\ast}_0({\cal{E}})=\frac{\frac{\alpha}{\beta}}{1+\sqrt{1+\frac{c_2\alpha}{c_1\beta}}}
\end{equation}
and $n=|E|.$
\end{lemma}

\begin{proof}
First, time complexity of the procedure step 1 is to be estimated.
It can be implemented in the straightforward way as follows.
Initially, set ${\cal{P}}:={\cal{N}}_{\varepsilon}$ and ${\cal{I}}:=\varnothing.$
An arbitrary set $P\in {\cal{P}}$ is tried for adding to ${\cal{I}}$
by performing a sequence of checks to find out if there is an object $I\in {\cal{I}}$ with $w(P\cap I)>\delta w(Y_0).$
Each check is done by running through all points from $Y_0$ in $O(n^2)$ time.
If such $I$ exists, then choose the first encountered one, add $P$ into a set ${\cal{N}}_{\delta,I}$ which is initially assumed empty
and set ${\cal{P}}:={\cal{P}}\backslash\{P\}.$
Otherwise, add $P$ into ${\cal{I}}$ and set ${\cal{P}}:={\cal{P}}\backslash\{P\}.$
The process stops when ${\cal{P}}=\varnothing.$

Let $t_{\theta_0}=|{\cal{I}}|$ and $z_i$ be the number of sets from ${\cal{N}}_{\varepsilon}$
which are tried for inclusion into ${\cal{I}}$ when
$|{\cal{I}}|=i.$ Then, time complexity of the above straightforward implementation is of the order:
$$O\left(n^2\sum\limits_{i=1}^{t_{\theta_0}}z_ii\right)=O\left(n^2 t_{\theta_0}\sum\limits_{i=1}^{t_{\theta_0}}z_i\right)=O(n^3
t_{\theta_0}).$$ Its space cost
is obviously of the same order as the space cost required to store the range space $(Y_0,{\cal{N}},w).$

Second, complexity of the step 2 is to be estimated of the {\sc Weak Epsilon Net Finder} procedure.
Here
one has disjoint sets ${\cal{N}}_{{\delta},I}$ formed for each
$I\in{\cal{I}}.$
It requires $O(Q(|{\cal{N}}_{\varepsilon}|))$
time by our assumption on time complexity of the subspace epsilon net finder,
working at the procedure step 2.

Finally, the claimed upper bound is established for length of a weak epsilon net produced by the {\sc Weak Epsilon Net Finder} procedure.
Let ${\cal{I}}$ be a maximal $\delta$-independent set, where the parameter $\theta_0$ is to be
chosen later. Following the same argument as in the proof of the theorem 4 from \cite{pyrga},
one gets $$t_{\theta_0}\leq\frac{\sum\limits_{I\in {\cal{I}}}w(I)}{\varepsilon w(Y_0)}=
\frac{\sum\limits_{y\in Y_0\cap\bigcup\limits_{I\in {\cal{I}}}I} w(y)n_{{\cal{I}}}(y)}{\varepsilon w(Y_0)}\leq$$$$\leq \frac{\sum\limits_{y\in Y_0\cap\bigcup\limits_{I\in {\cal{I}}}I} w(y)(m_{{\cal{I}}}(y)+\tau)}{\varepsilon\alpha w(Y_0)}=$$
$$=\frac{\tau w\left(Y_0\cap\bigcup\limits_{I\in {\cal{I}}}I\right)+\sum\limits_{u\in U}w(Y_0(u))}{\varepsilon\alpha w(Y_0)}\leq\frac{\tau w(Y_0)+t_{\theta_0}\beta\delta w(Y_0)}{\varepsilon\alpha w(Y_0)}=\frac{t_{\theta_0}\beta\theta_0}{\alpha}+\frac{\tau}{\alpha\varepsilon}$$
where $Y_0(u)\subset Y_0$ contains all points which lie in the intersection of a pair of those objects from ${\cal{I}}$
which form an edge $u\in U.$ Thus, it gives upper bounds $t_{\theta_0}\leq\frac{\tau}{(\alpha-\theta_0\beta)\varepsilon}$ and
$\frac{\sum\limits_{I\in {\cal{I}}}w(I)}{\varepsilon w(Y_0)}\leq\frac{\tau}{(\alpha-\theta_0\beta)\varepsilon}.$
According to our assumptions, the subspace epsilon net finder at the procedure step 2
gives a hitting set of size at most $\frac{c_1w(I)}{\delta w(Y_0)}+c_2$
for ${\cal{N}}_{{\delta},I},\,I\in{\cal{I}}.$
Therefore, $C_{\theta_0}$ is a weak $\varepsilon$-net
of size at most $\left(\frac{c_1}{\theta_0}+c_2\right)\frac{\tau}{(\alpha-\theta_0\beta)\varepsilon}.$ Optimizing with respect to $\theta_0<\frac{\alpha}{\beta}$
 one obtains $\theta_0^{\ast}=\theta_0^{\ast}({\cal{E}})=\frac{\frac{\alpha}{\beta}}{1+\sqrt{1+\frac{c_2\alpha}{c_1\beta}}}$ and gets the claimed bound $$|C_{\theta_0^{\ast}}|\leq \left[\left(1+\frac{1}{\sqrt{1+\frac{c_2\alpha}{c_1\beta}}}\right)\left(\frac{2c_1\tau\beta}{\alpha^2}+\frac{c_2\tau}{\alpha}\right)+\frac{c_2\tau}{\alpha \sqrt{1+\frac{c_2\alpha}{c_1\beta}}}\right]\frac{1}{\varepsilon}.$$
\end{proof}

\begin{definition}
Given $r>0$ and a class ${\cal{E}}\subseteq{\cal{E}}_0$ of sets of straight line segments on the plane,
a map is called a {\it{structural map}} for ${\cal{E}}$ and $r,$ if it
assigns a graph $G_{\cal{I}}$ for each $E\in{\cal{E}}$ and ${\cal{I}}\subseteq{\cal{N}}_r(E)$
as defined in the lemma \ref{dskskskkroroeowoqp},
where the constants $\alpha=\alpha({\cal{E}},r),\beta=\beta({\cal{E}},r)$ and $\tau=\tau({\cal{E}},r)$ are referred to as {\it structural parameters}
for the class ${\cal{E}}$ and radius $r.$
\end{definition}

Thus, to estimate value of $\theta_0$ for a given class ${\cal{E}}$ of sets of straight line segments on the plane, the approach of the lemma \ref{dskskskkroroeowoqp} suggests to choose an appropriate subspace epsilon net finder to use at the step 2 of the
{\sc Weak Epsilon Net Finder} procedure and identify a structural map for ${\cal{E}}$ and $r.$ More specifically, $\theta_0=\theta^{\ast}_0$ is computed using the equation $(\ref{hjhiislsldlsls}),$ where $\alpha,\beta$ and $\tau$ are parameters of the identified structural map whereas $c_1$ and $c_2$ are performance parameters of the chosen subspace epsilon net finder.

Let $M^{\ast}=\left(1+\frac{1}{\sqrt{1+\frac{c_2\alpha}{c_1\beta}}}\right)\left(\frac{2c_1\tau\beta}{\alpha^2}+\frac{c_2\tau}{\alpha}\right)+\frac{c_2\tau}{\alpha \sqrt{1+\frac{c_2\alpha}{c_1\beta}}}$ be the value of the performance parameter of the {\sc Weak Epsilon Net Finder} procedure,
which corresponds to setting $\theta_0=\theta^{\ast}_0.$
Let us note that $\theta^{\ast}_0$ depends only on ratios $\frac{c_1}{c_2}$ and $\frac{\beta}{\alpha}$ whereas $M^{\ast}$
depends on $c_1, c_2, \frac{\beta}{\alpha}$ and $\frac{\tau}{\alpha}.$
The smaller the latter two ratios, the smaller $M^{\ast}$ is.

Of course, the approach of the lemma \ref{dskskskkroroeowoqp}
to estimating $\theta_0$ can also be adapted for general
setting of the {\sc Hitting Set} problem for a range space
$(Y,{\cal{R}}),$ where $Y\subset\mathbb{R}^2$ is a
finite set and ${\cal{R}}$ is a family of subsets on the plane.
It is shown in the subsubsection \ref{skskfkgeoqoapdlgls} below that, being applied for
subspaces of the range space $(Y_0,{\cal{N}}_r(E),w),$
this approach provides the smaller performance parameter $M$
than that parameter for the epsilon net finder, resulting from
direct application of the original approach of \cite{pyrga}.

\subsubsection{Identifying a structural map and estimating $\theta_0$}\label{yuieruiusdjk}

Let ${\cal{E}}\subseteq{\cal{E}}_0.$ To build a structural map for ${\cal{E}}$ and $r$
a graph is to be identified for each $E\in{\cal{E}}$ and ${\cal{I}}\subseteq{\cal{N}}_r(E)$
or, equivalently, for each $E\in{\cal{E}},$ $E'\subseteq E$ and $r>0.$ Below it is shown that a map, which assigns
a Delaunay triangulation graph of $E'$ for every segment set $E',$
turns out to be a favourable structural map for which ratios $\frac{\beta}{\alpha}$ and $\frac{\tau}{\alpha}$
are small.

Delaunay triangulations can be defined \cite{breivellers} for planar segment sets
of non-overlapping straight line segments in assumption of their {\it general position}:
\begin{enumerate}
\item no quadruple exists of segments from $E$ which is touched by any single disk;

\item the set is in general position of endpoints of segments from $E.$
\end{enumerate}

\begin{definition}
Let $F$ be the maximal set of open non-overlapping triangles each of which has its endpoints
lying on 3 distinct segments from $E$ and its open circumscribing disk does not intersect any
segment from $E.$ The complement $${\rm{conv}}\,\left(\bigcup\limits_{e\in E}e\right)\backslash \left(\bigcup\limits_{f\in F}f\cup \bigcup\limits_{e\in E}e\right)$$
is a union of a set $U$ of relatively open connected components, where closure of each component intersects exactly
two segments from $E.$ A triple $T_E=(E,U,F)$ is called a Delaunay triangulation of the segment set $E.$
A graph $G_E=(E,U_2)$ is called a graph for $T_E,$ where $U_2$ consists of those unordered pairs $e,e'\in
E$ for which there exists $u\in U$ with $e\cap {\rm{cl}}\,u\neq\varnothing$ and $e'\cap {\rm{cl}}\,u\neq\varnothing.$
\end{definition}
It is shown in the section 4 of \cite{breivellers} that a Delaunay triangulation $T_E$ is
uniquely defined by a set $E$ of non-overlapping segments in general position.
Moreover, its graph $G_E$ is planar and dual to the graph of Voronoi diagram for $E.$

Below a parameter $\sigma=\sigma({\cal{E}})$ is given on which $\beta=\beta({\cal{E}})$ depends.
Let $m(E')=\left|\left\{e\in E':e\cap {\rm{bd}}\,{\rm{conv}}\,\left(\bigcup\limits_{e\in E'}e\right)\neq\varnothing\right\}\right|$ for $E'\subseteq E.$
Let also $$\sigma=\sigma({\cal{E}})=\inf\limits_{E'\subseteq E, E\in{\cal{E}}}\frac{m(E')}{|E'|}.$$

\begin{lemma}\label{ks234321345}
For any class ${\cal{E}}\subseteq{\cal{E}}_0$ and $r>0$ there exists a
structural map with $\beta=3-\sigma$
and $\alpha=\tau=1.$
\end{lemma}
\begin{proof}
It can be assumed that $E\in{\cal{E}}$ contains
pairwise non-intersecting segments. Indeed, the
{\sc PEH} problem can be considered for the same segment set $E$ with radius $r+\rho$ instead of $r,$
where a small constant $\rho>0$ guarantees
meeting the following conditions:
\begin{enumerate}
\item $\{N\in{\cal{N}}_r(E):y\in N\}=\{N\in{\cal{N}}_{r+\rho}(E):y\in N\}$
for every $y\in Y_0;$

\item a subset of ${\cal{N}}_r(E)$ has empty intersection iff the respective subset of
${\cal{N}}_{r+\rho}(E)$ has no common points.
\end{enumerate}
Then each segment $[v_1,v_2]\in E$ is replaced by the segment $[v_1+\kappa(v_2-v_1),v_2-\kappa(v_2-v_1)]$
(denote the set of segments thus obtained by $E_{\kappa}),$ where a small $\kappa=\kappa(v_1,v_2)>0$ guarantees
that the same conditions are met for $N_{r+\rho}\left(E_{\kappa}\right)$ instead of ${\cal{N}}_{r+\rho}(E),$
$y\in{\rm{int}}\,\bigcap\limits_{N\in {\cal{N}}_{r+\rho}(E_{\kappa}):y\in N}N$ for each $y\in Y_0$
and the parameter $\sigma$ is kept unchanged.
Suppose a structural map is defined for $E_{\kappa}$ and $r+\rho$
with $\alpha=\tau=1,\beta=3-\sigma$
and a graph $G_{{\cal{I}}_{\rho\kappa}}$ corresponds to a subset ${\cal{I}}_{\rho\kappa}$ under this map, where
segments from $E_{\kappa}({\cal{I}}_{\rho\kappa})$ are shortened segments from $E({\cal{I}})$ for ${\cal{I}}\subseteq{\cal{N}}_r(E).$
It is obvious that the same graph can be assigned for the set ${\cal{I}}$
taking into account that $\{N\in{\cal{I}}:y\in N\}=\{N\in{\cal{I}}_{\rho\kappa}:y\in N\}$ for any $y\in Y_0.$
Thus, it is assumed with slight abuse of terminology
that $y\in{\rm{int}}\,\bigcap\limits_{N\in {\cal{I}}:y\in N}N$ for each $y\in Y_0$ and $E$ consists of non-overlapping segments.
Moreover, it can be assumed that the segment set $E$ is
in general position. To achieve this, segments from $E$ can be slightly shifted keeping unchanged values of structural parameters $\alpha,\tau$ and $\sigma$ without breaking empty/nonempty intersections of subsets of objects from ${\cal{N}}_r(E).$

Let ${\cal{I}}\subseteq {\cal{N}}_r(E)$ and $G_{{\cal{I}}}$ be the maximal graph which is obtained from
a Delaunay triangulation graph for $E({\cal{I}})$
by removing redundant multiple edges. Due to the theorem 3 from \cite{breivellers},
$G_{{\cal{I}}}$ contains at most $3|E({\cal{I}})|-k-3$ edges,
where $k$ denotes the number of those edges of ${\rm{conv}}\,\left(\bigcup\limits_{e\in E({\cal{I}})}e\right),$
which are not segments of $E({\cal{I}}).$ As segments from $E({\cal{I}})$ are non-intersecting,
$m(E({\cal{I}}))\leq k$ and $G_{{\cal{I}}}$ has at most $\beta|E({\cal{I}})|$ edges.

Let ${\cal{I}}(y)=\{I\in{\cal{I}}:y\in I\}$ and $G_{{\cal{I}}}(y)$ be a subgraph of $G_{{\cal{I}}}$ induced
by the subset $E({\cal{I}}(y))$ as its set of $n_{{\cal{I}}}(y)$ vertices.
Let us prove that $\alpha=\tau=1$ by induction on $n_{{\cal{I}}}(y)$ for every $y\in Y_0.$
The case $n_{{\cal{I}}}(y)=1$ is obvious. Let us assume that any graph $G_{{\cal{I}}}(y)$
with $n_{{\cal{I}}}(y)\leq k$ vertices contains at least $n_{{\cal{I}}}(y)-1$ edges (for any $r)$ and suppose that $n_{{\cal{I}}}(y)=k+1.$

By perturbing $y$ within ${\rm{int}}\,\bigcap\limits_{I\in{\cal{I}}(y)}I$
it can be achieved that segments from $E({\cal{I}}(y))$ are at distinct distances from $y.$
Besides, let $e_0(y)\in E({\cal{I}}(y))$ be the farthest (among segments of $E({\cal{I}}(y)))$ segment
from $y.$ Denote by $y_0$ Euclidean projection
of $y$ onto $e_0(y).$ There is a point $y_1\in [y,y_0]$ which is equidistant
from $e_0(y)$ and some segment $e(y)\in E({\cal{I}}(y))\backslash\{e_0(y)\}$
whereas none of segments of $E({\cal{I}})\backslash \{e_0(y),e(y)\}$ is within the distance $\|y_0-y_1\|_2$ from $y_1$
(again may be after a small perturbation of $y).$ Due to duality between Delaunay triangulations
and Voronoi diagrams considered over the same segment set (see the theorem 4 from \cite{breivellers}),
we get that $G_{{\cal{I}}}(y)$ contains an edge which connects $e_0(y)$ and $e(y).$
Obviously, each segment of $E({\cal{I}}(y))$ has nonempty intersection with the radius $r$ disk centered at $y.$
Let $\gamma>0$ be so small such that $r_0=\|y-y_0\|_2-\gamma$ radius disk centered at $y$ intersects all $n_{{\cal{I}}}(y)-1$ segments from
$E({\cal{I}}(y))\backslash\{e_0(y)\}.$
Let $G_{{\cal{I}}}(y,\gamma)$ be the subgraph of $G_{{\cal{I}}}(y)$
induced by segments of $E({\cal{I}}(y))\backslash\{e_0(y)\}.$
Applying inductive assumption, one gets that $G_{{\cal{I}}}(y,\gamma)$ has at least $n_{{\cal{I}}}(y)-2$ edges.
Thus, the graph $G_{{\cal{I}}}(y)$ contains at least $n_{{\cal{I}}}(y)-1$ edges.
\end{proof}

In the proof of the lemma \ref{ks234321345} a structural map is built for an arbitrary subclass
${\cal{E}}\subseteq{\cal{E}}_0.$ The parameter $\beta$ of this structural map depends on the ${\cal{E}}$-specific parameter
$\sigma.$ This allows to design
${\cal{E}}$-specific implementations of the {\sc Weak Epsilon Net Finder} procedure
with different values of $\theta^{\ast}_0,$ giving
the smaller value $M^{\ast}$ of the performance parameter $M$
than that value in the general case where ${\cal{E}}={\cal{E}}_0.$
Three examples are given below of choice of $\theta_0$ for different classes of sets of segments.
The first example is for the general case ${\cal{E}}={\cal{E}}_0.$ In this case the subspace epsilon net finder
from the subsubsection \ref{qyuuififis} is chosen to work at the step 2 of the {\sc Weak Epsilon Net Finder} procedure.
\smallskip

\noindent {\bf Example 1.} Using the equation $(\ref{hjhiislsldlsls})$ and the lemma \ref{ks234321345}
one gets $\theta^{\ast}_0({\cal{E}}_0)=\frac{1}{\left(3+\frac{\sqrt{39}}{2}\right)}\approx 0.163,$ where
$\alpha=\tau=1,\,\beta=3,\,c_1=8$ and $c_2=2.$
\smallskip

The second example is for a special proper subclass of the class ${\cal{E}}_0.$ It shows how $\theta^{\ast}_0$
is changed when $\frac{\beta}{\alpha}$ varies.

\begin{definition}
A plane graph $G=(V,E)$ is called a {\it generalized outerplane} if
$$e\cap{\rm{bd}}\,{\rm{conv}}\,\left(\bigcup\limits_{e\in E}e\right)\neq\varnothing$$
for any $e\in E.$
\end{definition}

When each segment from $E$ has its both endpoints on the boundary of
${\rm{conv}}\,\left(\bigcup\limits_{e\in E}e\right)$ such a plane graph $G=(V,E)$ is known as an outerplane graph.

The lemma below follows from the lemma \ref{ks234321345}.
\begin{lemma}\label{gakakakeoqoeofodl}
For the class ${\cal{E}}_1$ of edge sets of arbitrary generalized outerplane graphs
there is a structural map with $\alpha=\tau=\sigma=1$ and $\beta=2.$
\end{lemma}
\smallskip

\noindent {\bf Example 2.} Using the equation $(\ref{hjhiislsldlsls})$ one has
$\theta^{\ast}_0({\cal{E}}_1)=\frac{1}{\left(2+\frac{3}{\sqrt{2}}\right)}\approx 0.242,$
where $\alpha=\tau=1,\,\beta=2,\,c_1=8$ and $c_2=2.$
\smallskip

The third example illustrates how $\theta^{\ast}_0$ depends on
the ratio $\frac{c_1}{c_2}$ of performance parameters of subspace epsilon net finder,
chosen to work at the step 2 of the {\sc Weak Epsilon Net Finder} procedure.
Consider the class ${\cal{E}}_2(r)$ of sets of straight line segments in which
for any $E\in{\cal{E}}_2(r)$ either $d(e,e')>r$ or $d(e,e')=0$ for any distinct $e,e'\in E.$
For this class a suitable subspace epsilon net finder is the one from the subsubsection \ref{hkfksiwiaoaofohdov}.
\smallskip

\noindent {\bf Example 3.} In view of the equation $(\ref{hjhiislsldlsls})$
$\theta^{\ast}_0({\cal{E}}_2(r))=\frac{1}{3+3\sqrt{3}}\approx 0.122,$
where $\alpha=\tau=1,\,\beta=3,\,c_1=1$ and $c_2=6.$
\smallskip

\subsection{Performance analysis of the Weak Epsilon Net Finder procedure}\label{fikdksoeofosqp}

As an obvious consequence of lemmas \ref{kob11111111}, \ref{dskskskkroroeowoqp} and \ref{ks234321345} time complexity and space cost can be estimated of the procedure.
\begin{lemma}\label{hhsurigklflskieigodls}
For any $E\in {\cal{E}}_0$ the {\sc Weak Epsilon Net Finder} procedure works in $O\left(\frac{n^3}{\varepsilon}\right)$ time
and linear space with respect to the space cost to store $(Y_0,{\cal{N}},w)$ for ${\cal{N}}\subseteq{\cal{N}}_r(E)$ and $w:Y_0\rightarrow\mathbb{Q}_+.$
\end{lemma}

The lemma below summarizes on $O\left(\frac{1}{\varepsilon}\right)$ upper bounds on size of weak $\varepsilon$-nets, returned by the {\sc Weak Epsilon Net Finder} procedure.
It follows from lemmas \ref{kob11111111}, \ref{kkfkfkdkdksks}, \ref{fsjdkgowoqpspa}, \ref{dskskskkroroeowoqp}, \ref{ks234321345} and \ref{gakakakeoqoeofodl}.
\begin{lemma}\label{kdksksiritfoosos}
Given a range space $(Y_0,{\cal{N}},w)$ for ${\cal{N}}\subseteq{\cal{N}}_r(E)$ and $w:Y_0\rightarrow\mathbb{Q}_+,$ the {\sc Weak Epsilon Net Finder} procedure with a suitable subspace epsilon net finder, working at its step 2,
returns a weak $\varepsilon$-net for $(Y_0,{\cal{N}},w)$ of size at most $\frac{M}{\varepsilon},$ where
\begin{enumerate}
\item $M=50+52\sqrt{\frac{12}{13}}$ for $E$ being an edge set of a plane graph;

\item $M=34+24\sqrt{2}$ in the case where $E$ is an edge set of a generalized outerplane graph;

\item $M=12+6\sqrt{3}$ if each pair of distinct segments from $E$ is at
Euclidean distance either zero or more than $r$ from each other;

\item $M=\frac{144}{5}+32\sqrt{\frac{3}{5}}$ in the case where $E$ is an edge set of
any subgraph of a Delaunay triangulation;

\item $M=20+12\sqrt{2}$ in the case of $E,$ being an edge set of any subgraph of
a generalized outerplane Delaunay triangulation.
\end{enumerate}
\end{lemma}

\begin{remark}
If $E$ is an edge set of any subgraph of an outerplane graph,
the {\sc Weak Epsilon Net Finder} procedure can be slightly modified
to get much smaller $M$ than that parameter for the case of generalized outerplane graphs. This improvement consists in applying a special algorithm of constructing a maximal $\delta$-independent set ${\cal{I}}$ at its step 1. The algorithm is based on the fact that for every subset $E'\subseteq E$ a segment can be chosen from $E'$, being an edge of ${\rm{conv}}\,\bigcup\limits_{e\in E'}e.$ Here the {\sc Subspace Weak Epsilon Net Finder} procedure is used from the subsubsection \ref{qyuuififis} at the step 2 of the {\sc Weak Epsilon Net Finder} procedure to get hitting sets for ${\cal{N}}_{\delta,I}$ of size at most $\frac{4w(I)}{\delta w(Y_0)}.$
\end{remark}

\subsubsection{Comparing the performance parameter of the Weak Epsilon Net Finder procedure with related work}\label{skskfkgeoqoapdlgls}

As an alternative to the {\sc Weak Epsilon Net Finder} procedure consider
an epsilon net finder, resulting from direct application
of the original approach from \cite{pyrga} to subspaces of $(Y_0,{\cal{N}}_r(E),w).$
Within this procedure a partition is performed of ${\cal{N}}_{\varepsilon}$ into much smaller subsets than
it is done in the {\sc Weak Epsilon Net Finder} procedure. Namely, ${\cal{N}}_{\varepsilon}$ is first partitioned into groups
$${\cal{N}}^s_{\varepsilon}=\{N\in{\cal{N}}:2^{s+1}\varepsilon w(Y_0)\geq w(N)>2^s\varepsilon w(Y_0)\}$$
for $s=0,\ldots,\log_2\frac{1}{\varepsilon}-1;$ then each group ${\cal{N}}^s_{\varepsilon}$ is partitioned into
subsets in the same way as ${\cal{N}}_{\varepsilon}$ is at the step 1 of the {\sc Weak Epsilon Net Finder} procedure.
More precisely, the corresponding partition of ${\cal{N}}^s_{\varepsilon}$ is generated using some maximal
$\delta_s$-independent set ${\cal{I}}_s$ for $(Y_0,{\cal{N}}^s_{\varepsilon},w),$ where $\delta_s=\theta_02^s\varepsilon.$
Applying results from \cite{komlos}, a hitting set is generated of constant size for each element of the partition of
${\cal{N}}^s_{\varepsilon}$ for every $s.$ Omitting details,
the original approach of \cite{pyrga} gives an epsilon net finder whose performance parameter is equal to $M_1=\frac{16d\tau\beta\log\frac{4\beta}{\alpha}}{\alpha^2},$ where $d$ is VC-dimension of $(Y_0,{\cal{N}}_r(E)).$

For $c_2=0$ the equation $(\ref{hjhiislsldlsls})$ gives $\theta^{\ast}_0=\frac{\alpha}{2\beta}$ and the bound $|C_{\theta^{\ast}_0}|\leq\frac{M_2}{\varepsilon}$ follows from the lemma \ref{dskskskkroroeowoqp}
for $M_2=\frac{4c_1\tau\beta}{\alpha^2}.$
Therefore $M_2\leq M_1$ when $c_1\leq 4d\log\frac{4\beta}{\alpha}.$
Of course, the {\sc Subspace Weak Epsilon Net Finder} procedure from the subsubsection \ref{qyuuififis}
can be considered as having performance parameters $c_1=10$ and $c_2=0.$
As $d\geq 3,$ $\frac{M_1}{M_2}\geq\approx 3.$

\section{Computing a proper weighting in the Piercing Hippodromes algorithm}\label{pkgkslsls}

Apart from constructing weak epsilon nets, another important task is performed at the step 4 of the {\sc Piercing Hippodromes} algorithm.
Given ${\cal{N}}\subseteq{\cal{N}}_r(E),$ it consists in computing a proper weight map $w_f:Y_0\rightarrow\mathbb{Q}_+$
such that $w_f(N)=\Omega\left(\frac{w_f(Y_0)}{{\rm{OPT}}(Y_0,{\cal{N}})}\right)$ for all $N\in {\cal{N}}.$
In \cite{bronnimann} it is shown (in much more general
setting) by Bronnimann and Goodrich and later in  \cite{agarwal} by Agarwal and Pan
that such a map always exists and can be computed algorithmically.

In this section an iterative reweighting procedure is described to be applied at the step 4 of the {\sc Piercing Hippodromes} algorithm.
It is
a slightly modified version of the procedure, applied within the Agarwal-Pan algorithm from \cite{agarwal}. An analogous modification is used in \cite{bus2}.
To present work of the procedure, the idea of its original version from \cite{agarwal} is briefly reviewed first.

\subsection{Iterative reweighting procedure in the original Agarwal-Pan algorithm}

Being applied for range subspaces of $(Y_0,{\cal{N}}_r(E)),$ the iterative reweighting procedure of the original Agarwal-Pan algorithm
accepts a positive integer parameter $k$ and a range space $(Y_0,{\cal{N}})$ for ${\cal{N}}\subseteq{\cal{N}}_r(E)$
as its input.
It updates weights of points from $Y_0,$ trying
to increase ratios $\frac{w(N)}{w(Y_0)}$ uniformly for all $N\in{\cal{N}}$
to get them all above the threshold $\frac{1}{2ke}.$

The procedure works as follows.
Initially, $w(y):=w_0(y)$ for all $y\in Y_0,$ where $w_0$ is the unit weight map.
Within the procedure steps are repeated for objects from ${\cal{N}},$
which are called {\it weight updating} steps.
Weight updating steps are grouped in the so called rounds. Each round contains
at most $T=2k$ consecutive weight updating steps.
Within each round objects from ${\cal{N}}$ are processed one by one.
It means that after performing weight updating steps for a particular
object from ${\cal{N}}$ it is not processed later in the current round.

Given an object $N\in{\cal{N}}$ such that $w(N)\leq\frac{w(Y_0)}{2k},$
an update $w(y):=2w(y)$ is performed for each $y\in N\cap Y_0$
during a particular weight updating step for $N.$
Weight updating steps for $N$ continue to be performed until either $w(N)>\frac{w(Y_0)}{2k}$
or the total number of weight updating steps (including those steps
for previously processed objects) equals to $T$ in the current round.

When the current round is finished without
achieving the limit $T$ on the total number of performed weight updating steps,
a final weight map $w_k:=w$ on $Y_0,$ the value $\varepsilon_k=\frac{1}{2ke}$ are returned
and the procedure stops. In this case it can be proved \cite{agarwal} that $w_k(N)>\varepsilon_kw(Y_0)$ for all $N\in{\cal{N}}.$
Otherwise, if the current round contains $T$ consecutive weight updating steps,
the procedure either proceeds to the next round or stops depending on
whether the total number of performed rounds does not exceed $2\log\frac{|Y_0|}{k}+1.$
If the procedure stops, exceeding that limit, it reports that
no suitable weight map can be computed for a given $k.$

Ability of the original Agarwal-Pan algorithm to identify
a value $k=O({\rm{OPT}}),$ for which the iterative reweighting procedure is able
to compute the corresponding suitable weight map on $Y_0,$
relies on a basic observation. It can formulated as follows
for $(Y_0,{\cal{N}}):$ for any positive integer $k$ the above procedure finishes working,
performing less than $T$ steps in the final round if a $k$-element
hitting set $C\subseteq Y_0$ exists for ${\cal{N}};$
besides, no $k$-element hitting set exists for ${\cal{N}},$
when the procedure performs exactly $T$ weight updating steps in its final round.
This observation is based on the fact that $w(C)$ grows faster than
$w(Y_0)$ as weight updating steps are proceeding, thus, restricting
the number of those steps (see the lemma 2.1 from \cite{agarwal} for details).

\subsection{Modified iterative reweighting procedure}

The problem with the original iterative reweighting procedure from \cite{agarwal}
is that applying it at the step 4 of the {\sc Piercing Hippodromes} algorithm
only guarantees an upper bound on its
approximation factor of $4Me$ (see the procedure analysis in \cite{agarwal}),
where $M$ is the performance parameter of the {\sc Weak Epsilon Net Finder} procedure.
Below a slight modification is provided of the original
procedure to achieve a better upper bound, which is very close to $M.$
The only difference between this procedure and the procedure
of the original algorithm is that weights are only slightly modified of objects
from ${\cal{N}}$ during a particular weight updating step
by just multiplying them on the factor of $1+\lambda_1,$ where $\lambda_1$ is some small absolute constant to be
chosen later; moreover, the condition $w(N)>\frac{w(Y_0)}{\lambda k}$ is verified before performing each weight updating step for some
constant $\lambda>1,\,\lambda\approx 1.$

Pseudo-code of the procedure is given below. Let $\kappa=2\lambda-\lambda\lambda_1-2>0.$

\smallskip

\hrule

\smallskip

\noindent {\sc Iterative Reweighting.}
\smallskip

\hrule
\smallskip

\noindent {\bf Input:} a parameter $k$ and a range space $(Y_0,{\cal{N}})$ for some
${\cal{N}}=\{N_1,\ldots,N_m\}\subseteq{\cal{N}}_r(E);$

\noindent {\bf Output:} if a weight map $w_k=w_k(\cdot|Y_0,{\cal{N}})$
can be computed for which $w_k(N)>\varepsilon_kw_k(Y)$ for all $N\in{\cal{N}},$
where $\varepsilon_k\geq\frac{1}{\lambda ke^{2\lambda_1/\lambda}},$
it returns such $w_k$ and $\varepsilon_k;$ otherwise, it reports that no suitable weight map
can be computed for a given $k.$
\smallskip

\hrule
\begin{enumerate}

\item set $t:=1;$ // {\it round counter}

\item set $w:=w_0$ // set $w$ equal to the unit weight map

\item set $s:=0$ and $p:=1;$ // {\it counters for weight updating steps and objects processed in the current round}

\item verify the inequality
\begin{equation}\label{klsks}
w(N_p)\leq\frac{w(Y_0)}{\lambda k}
\end{equation}
and
if it is true, set $s:=s+1,$ $w(y):=w(y)(1+\lambda_1)$ for all
$y\in N_p\cap Y_0$ and continue repeating the step 4 while $s<2k$ and $(\ref{klsks})$
still holds;

\item for $s<2k$ examine whether the equality $p=m$ holds:
if it does, return $w_k:=w,\,\varepsilon_k:=\frac{1}{\lambda ke^{\lambda_1s/(\lambda k)}};$ otherwise, set $p:=p+1$ and go to step 4;

\item if $s=2k,$ check if $t>\frac{\lambda\ln(|Y_0|/k)}{\lambda_1\kappa}$ holds:
when it does, report that no suitable map can be computed for a given value of $k;$
otherwise, set $t:=t+1$ and go to step 3.

\end{enumerate}
\hrule
\smallskip

Finally note that computing of $w(N_p)$ and generating points
from $N_p\cap Y_0$ at the procedure step 4
is done in a straightforward way. For example, when reporting
points from $N_p\cap Y_0$ each point $y\in Y_0$ is verified if
it is contained in $N_p:$ if yes, the point $y$ is reported.

\section{Constant factor approximation algorithms for the PEH problem and their performances}\label{fhdhhsueurusiis}

In this section our main algorithmic results are formulated.
The first result is a constant factor approximation
for the {\sc PEH} problem on a set of $r$-hippodromes whose underlying straight line segments are allowed
to intersect at most at their endpoints. More accurate approximations are also provided
for special geometric configurations of segments.

\begin{theorem}\label{ksakldkekeksa}
Applying the {\sc Iterative reweighting} procedure at the step 4 of the
{\sc Piercing Hippodromes} algorithm and the {\sc Weak Epsilon Net Finder}
procedure with a suitable subspace epsilon net finder at the algorithm steps 2 and 7 for $\theta_0$
computed according to the equation $(\ref{hjhiislsldlsls}),$ an approximation algorithm can be obtained for the {\sc PEH}
problem such that for any small $\nu>0$ it works in
$O\left(\left(n^2+\frac{n\log n}{\nu^2}+\frac{\log n}{\nu^3}\right)n^2\log n\right)$ time,
$O\left(\frac{n^2\log n}{\nu}\right)$ space and provides an
\begin{enumerate}
\item $\left(50+52\sqrt{\frac{12}{13}}+\nu\right)$-approximate solution in the case where $E$ is an edge set of a plane graph;

\item $(34+24\sqrt{2}+\nu)$-approximate solution for $E,$ being an edge set of a generalized outerplane graph;

\item $\left(12+6\sqrt{3}+\nu\right)$-approximate solution if each pair of distinct segments from $E$ is at Euclidean distance either zero or more than $r$ from each other;

\item $\left(\frac{144}{5}+32\sqrt{\frac{3}{5}}+\nu\right)$-approximate solution in the case where $E$ is an edge set of any subgraph of a Delaunay triangulation;

\item $(20+12\sqrt{2}+\nu)$-approximate solution for $E,$ being an edge set of any subgraph of a generalized outerplane Delaunay triangulation.
\end{enumerate}

\end{theorem}
\begin{proof}
The proof is organized in two stages.

\noindent {\sc Stage 1.} Let ${\cal{E}}\subseteq{\cal{E}}_0.$ Suppose that the {\sc Weak Epsilon Net Finder}
procedure produces weak $\varepsilon$-nets for subspaces of $(Y_0,{\cal{N}}_r(E),w)$
of size at most $\frac{M}{\varepsilon}$ for any $E\in{\cal{E}}.$
At the first stage
it is proved that
for any small constants $\mu_0,\lambda_1>0$ and a constant $\lambda>1$
for which $\kappa=2\lambda-\lambda\lambda_1-2>0,$
the {\sc Piercing Hippodromes} algorithm
is $M(\mu_0+\lambda e^{2\lambda_1/\lambda})$-approximate,
works in
\begin{equation}\label{fjdjsjsjfuruwuyq}
O\left(\left(\frac{n^3\log n+\frac{n^2\log n}{\mu_0}}{\lambda_1\kappa}+
n^3{\rm{OPT}}\right)\log {\rm{OPT}}\right)
\end{equation}
time and $O\left(\frac{n^2\log n}{\kappa}\right)$ space.
Its proof is analogous to proofs of
lemmas 2.1 and 3.1 from \cite{agarwal} and the lemma 2.1 from \cite{bus2}.

First, an approximation ratio is estimated of the {\sc Piercing Hippodromes} algorithm.
For a given $k$ it is proved
that at most $\left\lceil\frac{2\lambda k\ln|Y_0|}{\lambda_1\kappa}\right\rceil$ weight
updates are done in the {\sc Iterative Reweighting} procedure at its step 4 for sets from ${\cal{N}}_k,$
summing over all rounds, under assumption that there is a $k$-element hitting set $S_k\subseteq Y_0$ for ${\cal{N}}_k.$
Indeed, let $w^F_k(Y_0)$ (respectively, $w^F_k(S_k))$ be the weight $w(Y_0)$
(respectively, be the weight $w(S_k))$ observed at the end of the final round in which the {\sc Iterative Reweighting} procedure stops.
Let $z_k$ be also the total number of times that weights are updated (i.e. multiplied by $1+\lambda_1)$
of objects from ${\cal{N}}_k$ in all rounds. The following inequality holds true:
\begin{equation}\label{kdhdhshshah}
w^F_k(Y_0)\leq |Y_0|\left(1+\frac{\lambda_1}{\lambda k}\right)^{z_k}.
\end{equation}
From the other hand, one gets $$\frac{w^F_k(S_k)}{k}=\frac{\sum\limits_{c\in S_k}(1+\lambda_1)^{z_k(c)}}{k}\geq
(1+\lambda_1)^{\sum\limits_{c\in S_k}z_k(c)/k}\geq (1+\lambda_1)^{z_k/k}$$
where $z_k(c)$ denotes the number of times that $w(c)$ is updated. As $w^F_k(S_k)\leq w^F_k(Y_0)$
the inequality holds true $$k(1+\lambda_1)^{z_k/k}\leq |Y_0|\left(1+\frac{\lambda_1}{\lambda k}\right)^{z_k}.$$
Resolving it with respect to $z_k,$ one gets $z_k\leq\frac{2\lambda k\ln(|Y_0|/k)}{\lambda_1\kappa}$\footnote{Here the doubled inequality $x-\frac{x^2}{2}\leq\ln(1+x)\leq x$ is used for small $x>0.$}.
Thus, once the inequality ${\rm{OPT}}(Y_0,{\cal{N}}_k)\leq k$ holds,
after at most $\left\lceil \frac{2\lambda k\ln(|Y_0|/k)}{\lambda_1\kappa}\right\rceil$
weight updates the {\sc Iterative Reweighting} procedure stops, returning $w_k$ and $\varepsilon_k.$
At the step 5 of the {\sc Piercing Hippodromes} algorithm
true and false values are explored of the flag variable
to localize ${\rm{OPT}}(Y_0,{\cal{N}}_{k_p}).$
The algorithm finally gets $k_f\leq {\rm{OPT}}(Y_0,{\cal{N}}_{k_f})$ and outputs the
corresponding weight map $w_{k_f}$ and the parameter
$\varepsilon_{k_f}\geq\frac{1}{\lambda k_fe^{2\lambda_1/\lambda}}.$

Let $C\subset\mathbb{R}^2$ be the set of size at most $Mk_f\left(\mu_0+\lambda e^{\lambda_1s/(\lambda k_f)}\right),$
which is returned at the step 8 of the {\sc Piercing Hippodromes} algorithm.
It is proved below that $C$ is a hitting set for ${\cal{N}}_r(E).$
Let $w^I_{k_f}(Y_0)$ be $w(Y_0)$ observed at the beginning of the (final) round
within which the {\sc Iterative Reweighting} procedure returns a proper weight map $w_{k_f}$
and the parameter $\varepsilon_{k_f}.$
As $s<2k_f$ in that round, one gets
$$w^F_{k_f}(Y_0)\leq\left(1+\frac{\lambda_1}{\lambda k_f}\right)^{s}w^I_{k_f}(Y_0)\leq e^{\lambda_1s/(\lambda k_f)}w^I_{k_f}(Y_0)<e^{2\lambda_1/\lambda}w^I_{k_f}(Y_0).$$
From the other hand, $w_{k_f}^F(N)>\frac{w^I_{k_f}(Y_0)}{\lambda k_f}\geq\frac{w^F_{k_f}(Y_0)}{\lambda k_fe^{\lambda_1s/(\lambda k_f)}}>\frac{w^F_{k_f}(Y_0)}{\lambda k_fe^{2\lambda_1/\lambda}}$ for any $N\in{\cal{N}}_{k_f},$
where $w_{k_f}^F(N)$ denotes weight of $N$ at the end of the final round.
Thus, the algorithm output $C$ gives a hitting set for ${\cal{N}}_r(E)$ and:
$$|C|\leq Mk_f\left(\mu_0+\lambda e^{\lambda_1s/(\lambda k_f)}\right)\leq M\left(\mu_0+\lambda e^{\lambda_1s/(\lambda k_f)}\right){\rm{OPT}}(Y_0,{\cal{N}}_{k_f})\leq $$ $$\leq M\left(\mu_0+\lambda e^{2\lambda_1/\lambda}\right){\rm{OPT}}(Y_0,{\cal{N}}_{k_f})\leq M\left(\mu_0+\lambda e^{2\lambda_1/\lambda}\right){\rm{OPT}}.$$

Now bounds are established for time complexity of the {\sc Piercing Hippodromes} algorithm.
Its step 1 requires $O(n^2)$ time.
Binary search over steps 2-4 requires
$O\left(\left(n^3{\rm{OPT}}+B\right)\log{\rm{OPT}}\right)$ time,
where $B$ is the maximal time complexity of the {\sc Iterative Reweighting} procedure. Indeed, the algorithm step 2
takes $O(n^3{\rm{OPT}})$ time by the lemma \ref{hhsurigklflskieigodls}
whereas its step 3 requires $O(n\,{\rm{OPT}})$ time.

The step 7 of the {\sc Piercing Hippodromes} algorithm takes
$O(n^3{\rm{OPT}})$ time by the lemma \ref{hhsurigklflskieigodls}.
Thus, to estimate total complexity of the algorithm it remains to estimate $B.$
Recall that the {\sc Iterative reweighting} procedure call is
for the space $(Y_0,{\cal{N}}_k)$ for $k=O({\rm{OPT}}).$
For any $t$ $t$th round consists of at most $|{\cal{N}}_k|$ operations
to compute weights of objects from ${\cal{N}}_k.$ As $t\leq \frac{\lambda\ln(|Y_0|/k)}{\lambda_1\kappa}+1,$ overall time complexity
is of the order $O\left(\frac{\lambda|{\cal{N}}_k|n^2\log n}{\lambda_1\kappa}\right)$ for
such operations at the procedure step 4. One has $|N\cap Y_0|\leq\frac{|Y_0|}{\mu_0 k}$ for every $N\in{\cal{N}}_k.$
As $s\leq 2k,$ an $O\left(\frac{\lambda n^2\log n}{\mu_0\lambda_1\kappa}+\frac{\lambda|{\cal{N}}_k|n^2\log n}{\lambda_1\kappa}\right)$
time is spent for operations to update point weights at the step 4. Therefore $B=O\left(\frac{\lambda \left(n^3\log n+\frac{n^2\log n}{\mu_0}\right)}{\lambda_1\kappa}\right).$

As for space cost of the {\sc Piercing Hippodromes} algorithm, note that $$w(Y_0)\leq |Y_0|^{1+2/\kappa}e^{2\lambda_1/\lambda}$$
substituting $z_k=\left(\frac{\lambda\ln (|Y_0|/k)}{\lambda_1\kappa}+1\right)2k$ into the bound $(\ref{kdhdhshshah})$ for $w^F_k(Y_0).$
Of course, as $w(y)=(1+\lambda_1)^{z_k(y)}$ for $y\in Y_0,$ it can be shown that
$$\sum\limits_{y\in Y_0}\ln w(y)\leq|Y_0|\ln\frac{w(Y_0)}{|Y_0|}=
O\left(\frac{|Y_0|\ln|Y_0|}{\kappa}\right).$$

\noindent {\sc Stage 2.} Now one is ready to prove the first statement of the theorem.
The other theorem statements can be proved analogously, using the lemma \ref{kdksksiritfoosos}.
Consider an implementation of the {\sc Weak Epsilon Net Finder} procedure with the {\sc Subspace Weak Epsilon Net Finder} procedure from the subsubsection \ref{qyuuififis}, working at its step 2. Choose $\theta_0$ as in the example 1 from the subsubsection \ref{yuieruiusdjk}.
By the lemma \ref{kdksksiritfoosos} the {\sc Weak Epsilon Net Finder} procedure outputs weak $\varepsilon$-nets
of size at most $\left(50+52\sqrt{\frac{12}{13}}\right)\frac{1}{\varepsilon}.$ Due to shown in the first stage,
it implies that the {\sc Piercing Hippodromes} algorithm is
$\left(50+52\sqrt{\frac{12}{13}}\right)(\mu_0+\lambda e^{2\lambda_1/\lambda})$-approximate.
One
can adjust its parameters $\mu_0,\lambda$ and $\lambda_1$ to get an $\left(50+52\sqrt{\frac{12}{13}}+\nu\right)$-approximate algorithm
without affecting orders of its time and space complexities
for any small $\nu>0.$ More specifically, setting $\lambda:=1+\frac{\nu}{300},\,\lambda_1=\mu_0:=\frac{\nu}{1200},$
it can be shown that $\left(50+52\sqrt{\frac{12}{13}}\right)\left(\frac{\nu}{1200}+\left(1+\frac{\nu}{300}\right)e^{\nu/600}\right)\leq
50+52\sqrt{\frac{12}{13}}+\nu$ for small $\nu>0.$ Moreover, it gives claimed dependencies
of the algorithm time and space cost on $\nu.$

\end{proof}

\section{Conclusion}

Constant factor approximations are proposed for a special NP- and W[1]-hard
geometric piercing problem on a set of $r$-hippodromes whose underlying straight line segments are allowed to intersect at most
at their endpoints. More accurate approximations are also provided for special
configurations of segments, forming edge sets of generalized outerplane graphs and Delaunay triangulations.
They demonstrate a competitive combination of guaranteed constant approximation
factor and time complexity being compared with known local search and epsilon net based
approximation algorithms for the {\sc Hitting Set} problem on a set of pseudo-disks.
We hope that our $O(1)$-approximations can be expedited by incorporating
clever geometric data structures.

\appendix

\section{Proof of the lemma \ref{gjfkdkksaadqqrsra}}

\begin{proof}
Let ${\cal{J}}=\{J_i\}_{i=1}^n$ be a set of bounded intervals on the real line, i.e. $J_i=[a_i,b_i],\,i=1,\ldots,n.$
Let ${\cal{J}}'$ be its subset of intervals
which is maximal with respect to inclusion and does not contain pairs of
intervals $I$ and $J$ with either $I\subseteq J$ or $J\subset I.$
Removing such pairs from ${\cal{J}}$ can be done in $O(n\log n)$ time and
$O(n)$ space. Indeed, an interval $[a,b]$ can be represented by a point $(a,b)$
on the $xy$-plane above the straight line $y=x;$ checking if an interval
$[a,b]$ contains some other interval $[c,d]$ is equivalent to checking
if the axis-parallel rectangle contains a point $(c,d)$
whose left upper vertex is $(a,b)$ and right lower vertex is $(b,a).$
This check can be done using data structures for processing of
orthogonal range emptiness queries on $n$-point sets in
$O(\log n)$ time and $O(n)$ space with preliminary preprocessing
in $O(n\log n)$ time \cite{Blelloch}.

Secondly, we get lower and upper ends of intervals
from ${\cal{J}}'$ sorted in a single sequence, set $H:=\varnothing$
and ${\cal{P}}:={\cal{J}}'.$ Then, doing sequentially until ${\cal{P}}=\varnothing,$
an interval $I_k=[a_k,b_k]\in {\cal{P}}$ is selected at step $k$
with the maximal upper end $b_k;$ its lower end $a_k$
is added to the hitting set $H$ and intervals are excluded from ${\cal{P}}$ which are hit by $a_k.$
When ${\cal{P}}=\varnothing,$ let $Q=\{I_k\}\subset{\cal{J}}$ be the set of non-overlapping intervals,
thus, constructed.
We get that $H$ is the minimum cardinality hitting set for ${\cal{J}}.$

Summarizing on the complexity of computing of $H,$ we note that
sorting of interval ends from ${\cal{J}}'$ can obviously be done in $O(|{\cal{J}}'|\log |{\cal{J}}'|)$ time. Moreover,
when reporting those intervals from ${\cal{P}},$
which contain $a_k,$
we first start with the interval from ${\cal{P}}$
having the second maximal upper end. Thus, it takes
$O(|{\cal{J}}'|)$ overall time for reporting such intervals.
\end{proof}

\end{document}